\documentclass{article}

\pdfoutput=1

     \usepackage[preprint, nonatbib]{neurips_2020}

\usepackage[utf8]{inputenc} 
\usepackage[T1]{fontenc}    
\usepackage{hyperref}       
\usepackage{url}            
\usepackage{booktabs}       
\usepackage{amsfonts}       
\usepackage{nicefrac}       
\usepackage{microtype}      

\usepackage{graphicx}
\usepackage{subfigure}
\usepackage{algorithm}
\usepackage{algorithmic}
\usepackage{amsmath}
\usepackage{capt-of}
\usepackage{wrapfig}
\usepackage{multirow}

\newtheorem{theorem}{Theorem}
\newtheorem{lemma}{Lemma}
\newtheorem{proof}{Proof}[section]

\title{Shuffle-Exchange Brings Faster: Reduce the Idle Time During Communication for Decentralized Neural Network Training}

\author{%
  Xiang~Yang * , Qi~Qi * , Jingyu~Wang * , Haifeng Sun * , Yujian Li * , Temp \dag \\
  * Beijing University of Posts and Telecommunications\\
  \texttt{\{yangxiang, qiqi8266, wangjingyu, hfsun, beiyouliyujian\}@bupt.edu.cn} \\
   \dag Singapore University of Technology and Design \\
  \texttt{tonyquek@sutd.edu.sg}
}

\begin{document}

\maketitle

\begin{abstract}
As a crucial scheme to accelerate the deep neural network (DNN) training, distributed stochastic gradient descent (DSGD) is widely adopted in many real-world applications. In most distributed deep learning (DL) frameworks, DSGD is implemented with Ring-AllReduce architecture (Ring-SGD) and uses a computation-communication overlap strategy to address the overhead of the massive communications required by DSGD. However, we observe that although $O(1)$ gradients are needed to be communicated per worker in Ring-SGD, the $O(n)$ handshakes required by Ring-SGD limits its usage when training with many workers or in high latency network. In this paper, we propose Shuffle-Exchange SGD (SESGD) to solve the dilemma of Ring-SGD. 
 In the cluster of 16 workers with 0.1ms Ethernet latency, SESGD can accelerate the DNN training to $1.7 \times$ without losing model accuracy. Moreover, the process can be accelerated up to $5\times$ in high latency networks (5ms).
\end{abstract}

\section{Introduction}

With the expansion of data and model scale, distributed stochastic gradient descent (DSGD) is widely adopted to accelerate the training of deep neural networks (DNNs). In DSGD, all parallel workers iteratively compute gradients with their local datasets and synchronize the gradients with others in every iteration to refine the global model parameters. In practice, most popular distributed deep learning (DL) frameworks implement DSGD with the Ring-AllReduce architecture and computation-communication overlap strategy to address the overhead of the massive communications in DSGD (e.g., TensorFlow \cite{tensorflow},  PyTorch\cite{pytorch}, Horovod \cite{sergeev2018horovod}).

Previous efforts reduce the impact of communication constraint in DSGD from two aspects. From the aspect of algorithm, Gradient Quantization \cite{quantization2,quantization3} and Sparsification \cite{sparse2,sparse3,dgc} reduce the communication time by cutting down the size of gradients in communication. 
These techniques need to balance the trade-off between the model accuracy and the size of gradients.
From the aspect of framework design, \cite{zhang2017poseidon} proposes a novel architecture, in which gradient communication is overlapped with backward computation layer by layer to reduce the execution time of these two operations. \cite{li2018pipe,shen2019faster,wang2019scalable} extend the overlap strategy by sending more layers or reducing communication frequency. 
These frameworks focus on overlapping computation and computation to achieve better throughput.

However, from a more fine-grained perspective for the communication cost, although $O(1)$ gradients are needed to be communicated per worker for DSGD with Ring-AllReduce architecture (Ring-SGD), $O(n)$ handshakes are needed when training with $n$ workers. In a cluster with bandwidth $\nu$ and network latency $t_\tau$, the communication overhead of Ring-SGD follows the formula on the basis of analysis in Section \ref{sec_allreduce}:

\begin{equation} 
T_{layer} \approx \frac{2\mathcal{G}}{\nu}+2nt_{\tau}
\end{equation}

where $T_{layer}$ is the communication overhead to synchronize gradients for a layer among all parallel workers, and $\mathcal{G}$ is transferred gradients.

It is expensive to deploy a training cluster with high bandwidth and low latency (e.g., Nvlink, InfiniBand). And for some specific applications high latency is unavoidable  \cite{45648}. When training DNN with many workers or in networks with high latency, the idle time $2nt_{\tau}$ caused by the frequent handshakes of Ring-SGD sharply increases and dominates the communication, which slows the training process and makes previous efforts less effective.  

In this paper, we propose and design \emph{Shuffle-Exchange SGD} (SESGD), a novel decentralized communication approach to solve $O(n)$ handshakes of Ring-SGD. In SESGD, all the workers are divided into different groups. In the beginning of every iteration, workers will be re-assigned to different groups through the Shuffle-Exchange operation. Instead of synchronizing the gradients among all workers, SESGD only synchronizes gradients within groups. Workers in different groups do not communicate in a single iteration, but different workers may be shuffled into the same group in the following iterations. 

In a nutshell, we make the following contributions:

\begin{enumerate}
	\item We observe the idle time caused by frequent handshakes during communication also makes distributed training inefficient.
	\item We propose SESGD, a novel decentralized communication approach to reduce the idle time during communication and guarantee the convergence.
	\item When training deep neural networks, SESGD achieves a better time speedup comparing with existing distributed algorithms.
\end{enumerate}

We implemented SESGD and evaluated our method in various datasets. Our experiments show that, in the cluster of 16 workers with 0.1ms Ethernet latency, SESGD can accelerate the entire training process to $1.7 \times$ without losing model performance. In high latency network, the process can be accelerated up to $5\times$.

\section{Background And Motivation}
In this section, we briefly introduce the basic concepts of distributed training.
Then we analyze the overhead of frequent handshakes caused by Ring-SGD.

\subsection{Computation and Communication Overlap}

In traditional DSGD, each worker conducts computation and communication operation sequentially in one iteration, as shown in Figure \ref{traditional}. During this progress, the GPUs are unused during communication. Since the computation of gradients is performed through DNN layer by layer, previous efforts propose to overlap the computation and communication by sending gradients layer by layer,
as shown in Figure \ref{overlaped_training}. This technique improves the throughput in DSGD, and is adopted and implemented as the default choice by most popular distributed DL frameworks
(e.g., TensorFlow, PyTorch, MXNet \cite{mxnet}).

\subsection{Decentralized Ring-AllReduce Architecture} \label{sec_allreduce}


\begin{figure}[tb]
	\centering
	\subfigure[training with traditional strategy] { \label{traditional}
		\centering
		\includegraphics[width=0.35\linewidth]{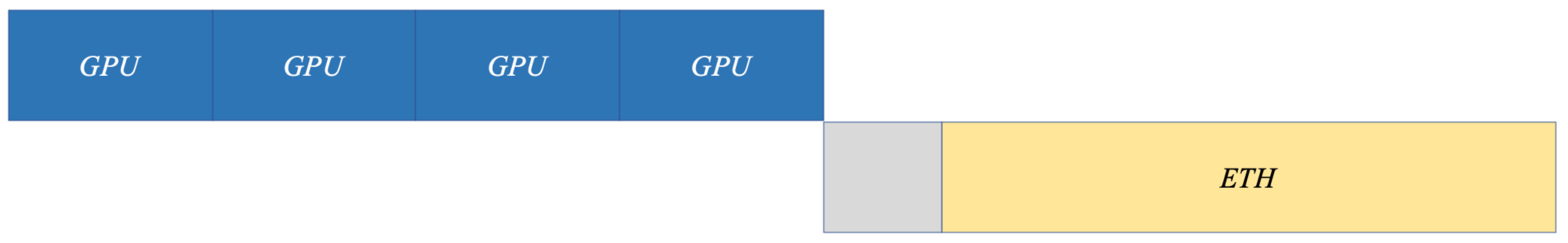}
	}
	\subfigure[training with overlap strategy] { \label{overlaped_training}
		\centering
		\includegraphics[width=0.3\linewidth]{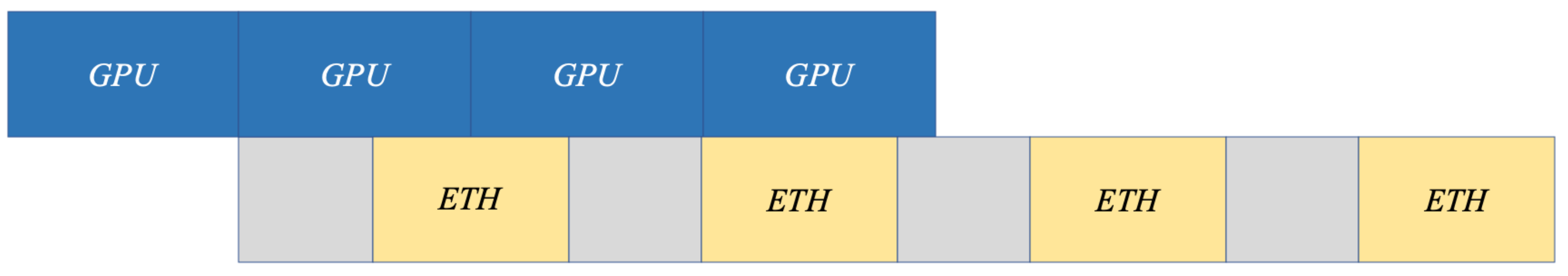}
	}
	\subfigure[training with SESGD]{
		\centering
		\includegraphics[width=0.25\linewidth]{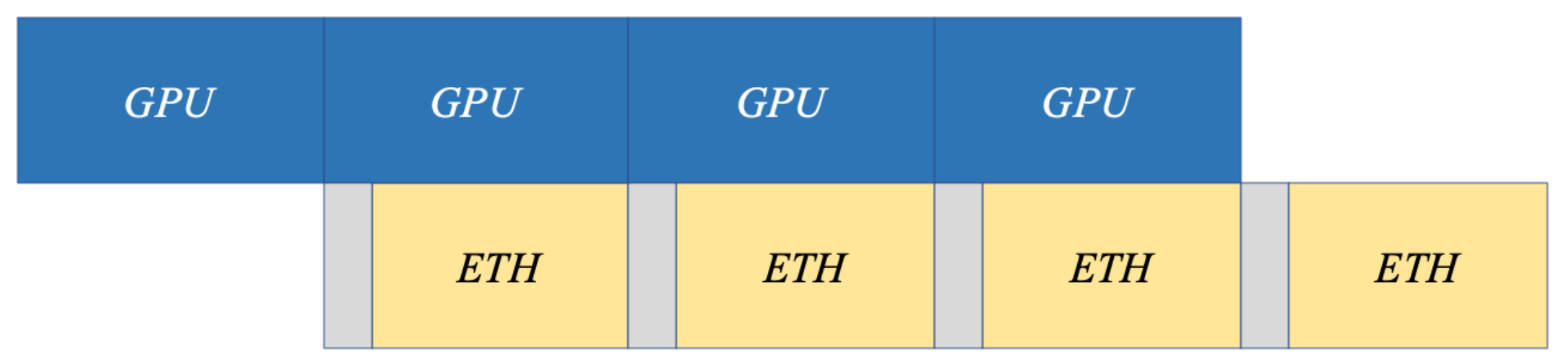}
	}
	\caption{Training with different strategies. A blue block indicates computation time in one layer,
		which uses the GPU resources. A yellow block corresponds to gradient transmission time in communication, which uses the bandwidth resources. And the grey block presents the idle time during communication, which harms the bandwidth utilization.}
	\label{fg_overlap}
\end{figure}

\begin{figure}[ht]
	\centering
	\begin{minipage}{0.3\linewidth}
		\centering
		\subfigure[Centralized communication] { \label{fg_ps}
			\centering
			\includegraphics[width=\linewidth]{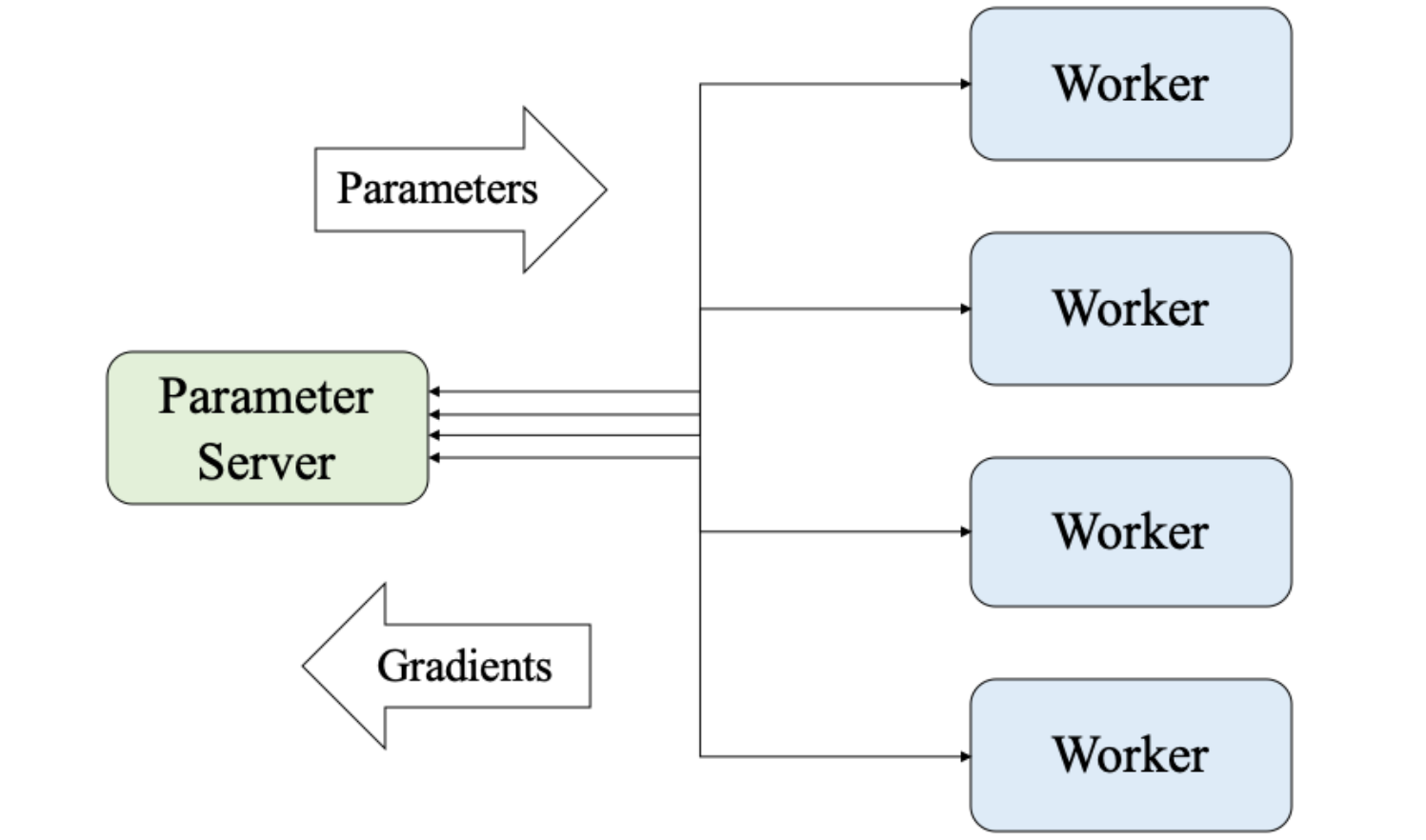}
		}
		
		\subfigure[Decentralized communication] { \label{fg_ring}
			\centering
			\includegraphics[width=\linewidth]{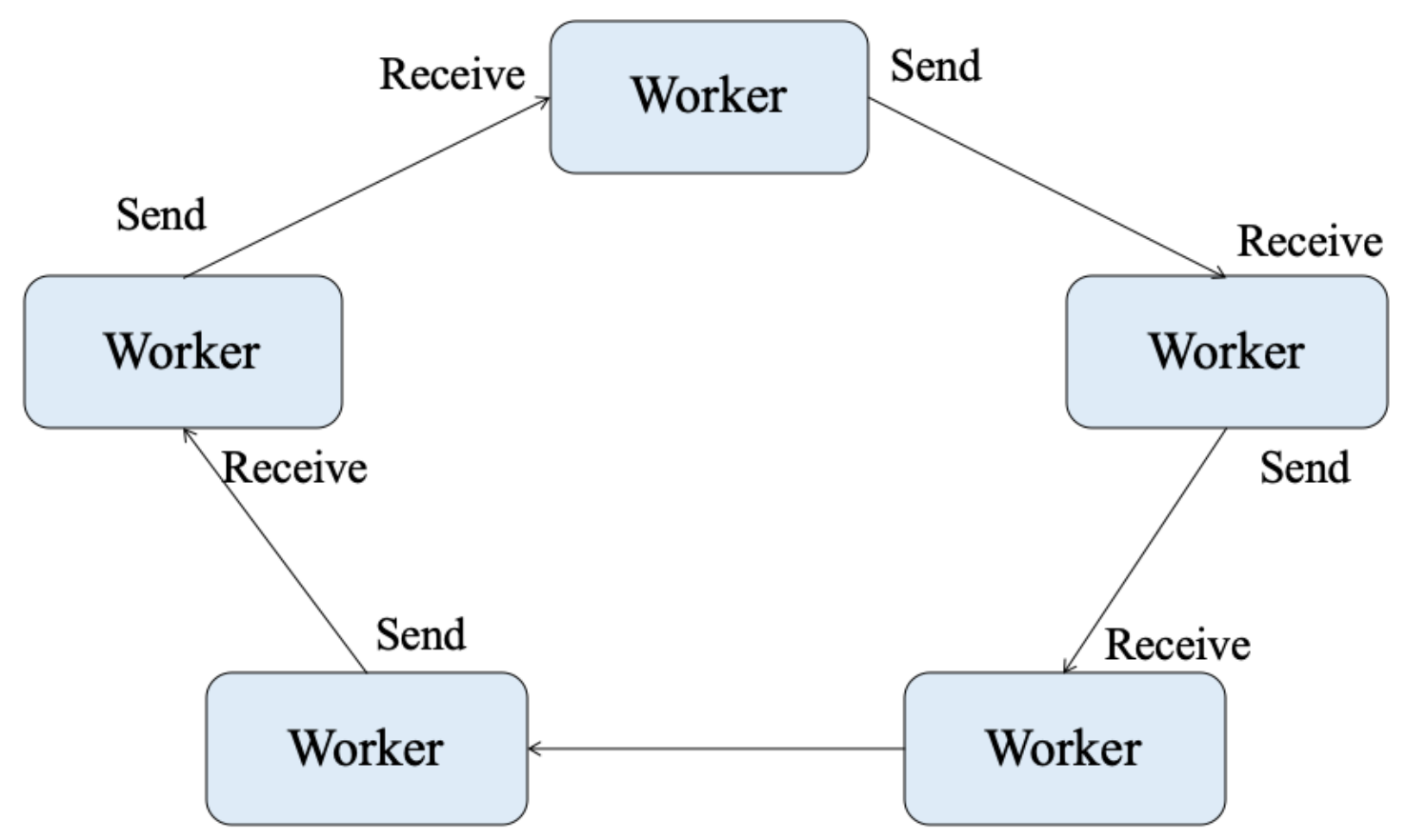}
		}
	\end{minipage}
	\begin{minipage}{0.65\linewidth}
	\subfigure[The schematic diagram of Ring-AllReduce]{ \label{fg_allreduce}
		\centering
		\includegraphics[width=\linewidth]{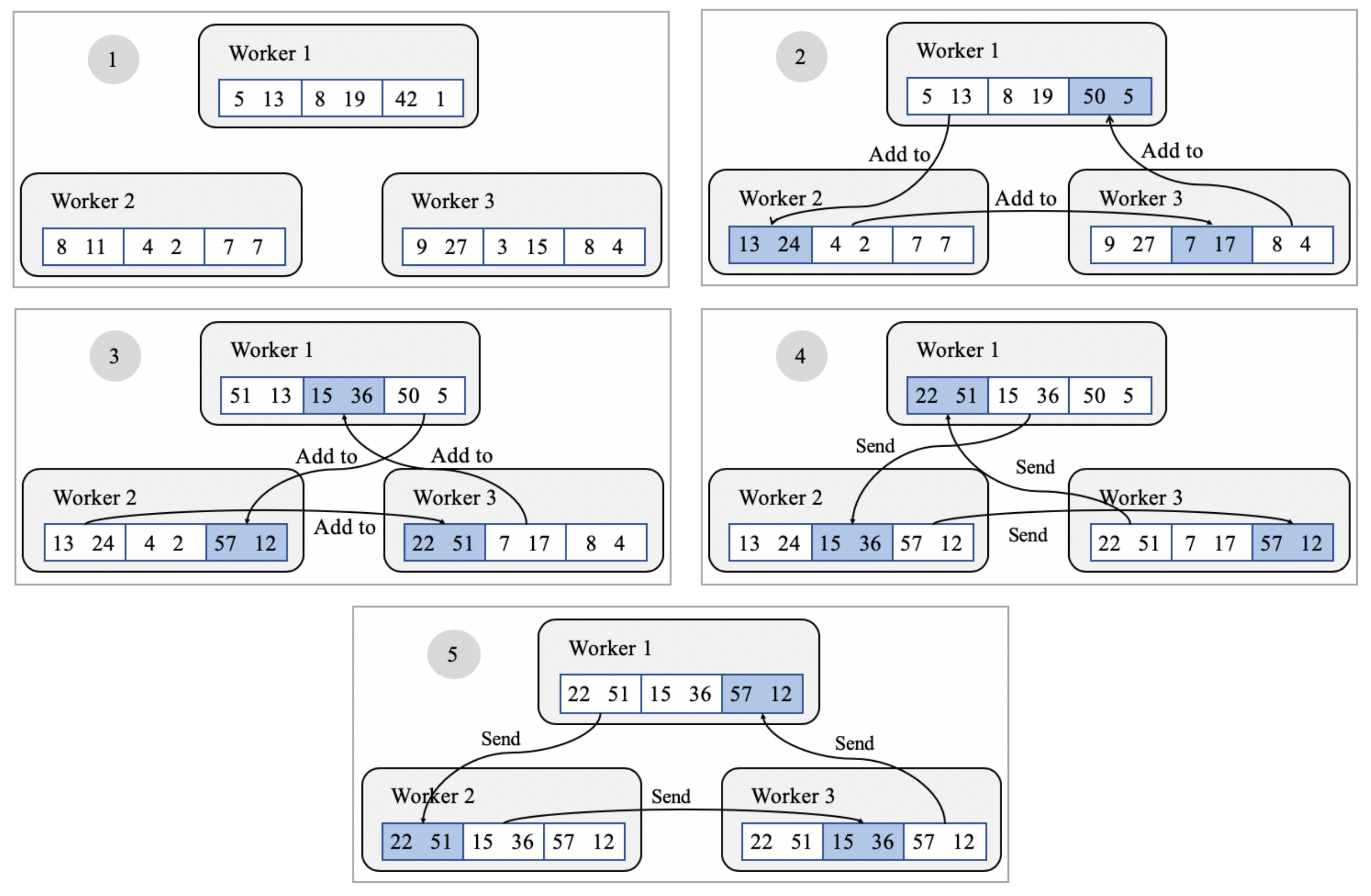}
	}
	\end{minipage}

	\caption{Illustrations of different communication architectures.}
\end{figure}


Most distributed DL frameworks employ the Ring-AllReduce architecture to train DL models iteratively. In this architecture, workers are arranged in a logical ring, for one communication with  as shown in Figure \ref{fg_ring}. Communication under Ring-AllReduce has two steps: Scatter-Reduce and All-Gather. Each requires $n-1$ handshakes. Each worker divides its computed gradient into $n$ slices. In one handshake, each worker sends one slice to its right neighbor while receives slice from its left. For one communication with 3 workers, it needs $2*(3-1)=4$ handshakes to get the sum of gradients, as shown in Figure \ref{fg_allreduce}.  The communication overhead can be expressed quantitatively as:

\begin{equation} \label{eq_allreduce}
T_{layer} = 2(n-1)(\frac{\mathcal{G}}{n\nu}+t_{\tau})
 \approx \frac{2\mathcal{G}}{\nu}+2nt_{\tau}
\end{equation}

Compared with the traditional Parameter Server architecture \cite{ps_arch} in Figure \ref{fg_ps}, where the parameter server needs to average $O(n)$ gradients and send the parameters to all workers, Ring-AllReduce as a decentralized architecture can guarantee $O(1)$ gradient transmission for every worker. The Parameter Server architecture holds the sending and receiving operations as two time steps during communication. Meanwhile, workers in Ring-AllReduce architecture can send and receive gradients at the same time, taking full use of the advantage of the full-duplex link. However, the $O(n)$ handshakes remains a problem.

\subsection{Overhead of Frequent Handshakes}

Most previous efforts focus on cutting down the size of the gradients or improving the overlaps between computation and communication, but pay little attention to the idle time of bandwidth during communication. Ring-SGD requires $O(n)$ handshakes to synchronize one layer. Therefore, the communication cost grows as the number of workers increase. When training with a large number of devices or the network latency is high, the $O(n)$ handshakes required by Ring-AllReduce dominates the communication. At the same time, the large number of small gradient pieces split by each worker will also greatly harm bandwidth utilization, make the situation even worse, as shown in Figure \ref{fg_band}.

We evaluate four different models with different number of workers, and measure the proportion of idle time during communication. The result is shown in Table \ref{tb_comm_time}. The idle time makes Ring-SGD with overlap strategy cannot be effectively applied with a large number of devices. For example, when training DNN with 50 layers on 16 workers, $50*2*(16-1) = 1500$ handshakes establish in one iteration. Previous efforts show that, with the same size, the deeper model with smaller layers often performs better \cite{szegedy2015going,szegedy2016rethinking}. The idle time during communication may be a bottleneck for most real-world applications.

\begin{figure}
	\centering
	\begin{minipage}{0.6\linewidth}
			\centering
			\begin{tabular}{lrrrr}
				\toprule
				Model       & Worker & Idle Time(s) & Proportion \\
				\midrule
				\multirow{2}{*}{ResNet18} 
				 & 4     & 0.384      & 30.27\%  \\
				& 16    & 1.326      & 61.02\%  \\
				\midrule
				\multirow{2}{*}{ResNet50} 
				    & 4     & 0.611      & 38.54\%  \\
				    & 16    & 2.910      & 69.02\%  \\
				\midrule
				\multirow{2}{*}{DenseNet121} 
				 & 4     & 1.244      & 57.66\%  \\
				 & 16    & 6.568      & 87.72\%  \\
				\midrule
				\multirow{2}{*}{VGG16} 
				       & 4     & 0.547      & 12.21\%  \\
				       & 16    & 1.722      & 38.44\%  \\
				\bottomrule
			\end{tabular}
			\captionof{table}{The idle time caused by $O(n)$ handshakes quickly grows for various DNNs with 0.1ms latency as the number of workers increases.}
			\label{tb_comm_time}
	\end{minipage}
	\begin{minipage}{0.3\linewidth}
	\centering
	\subfigure [Training with Ring-SGD]{
		\centering
		\includegraphics[width=\linewidth]{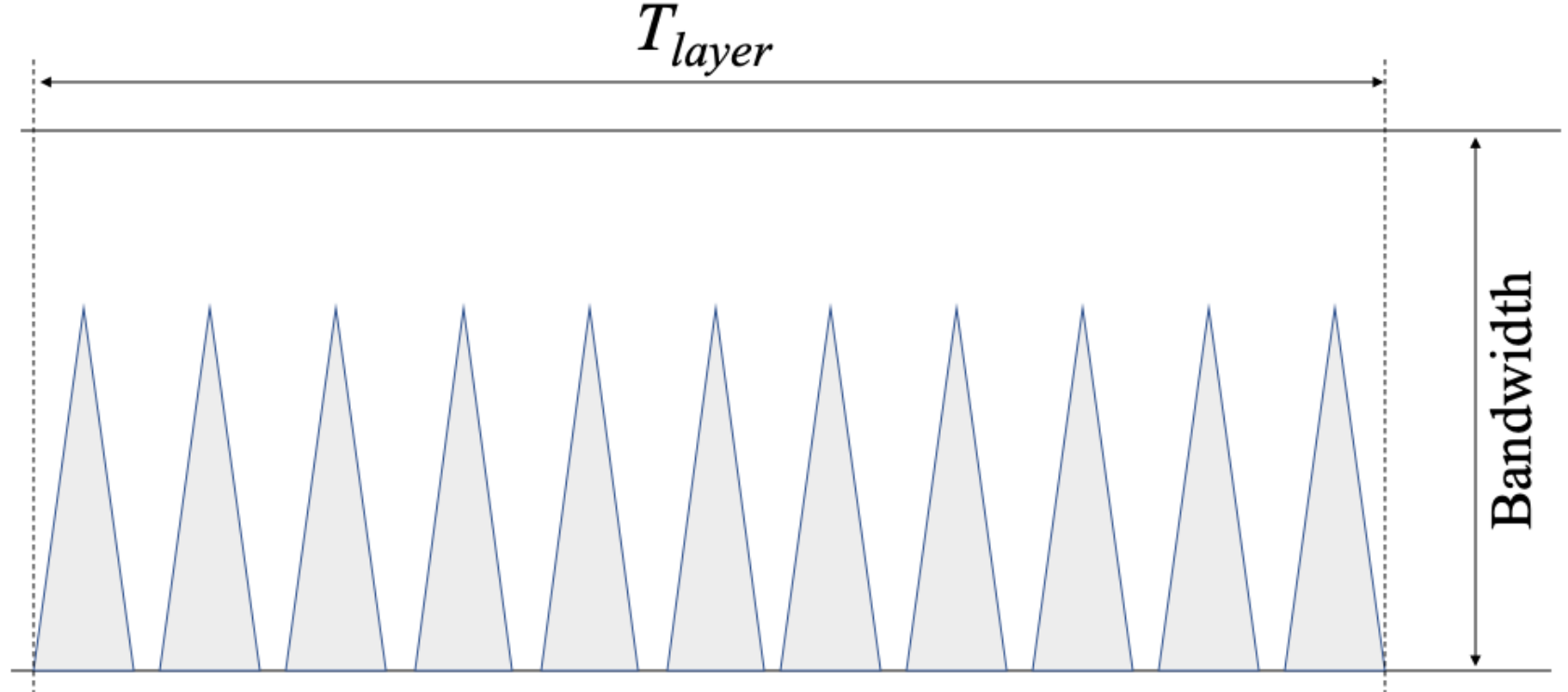}
	}
	\centering
	\subfigure [Training with SESGD] {
		\centering
		\includegraphics[width=\linewidth]{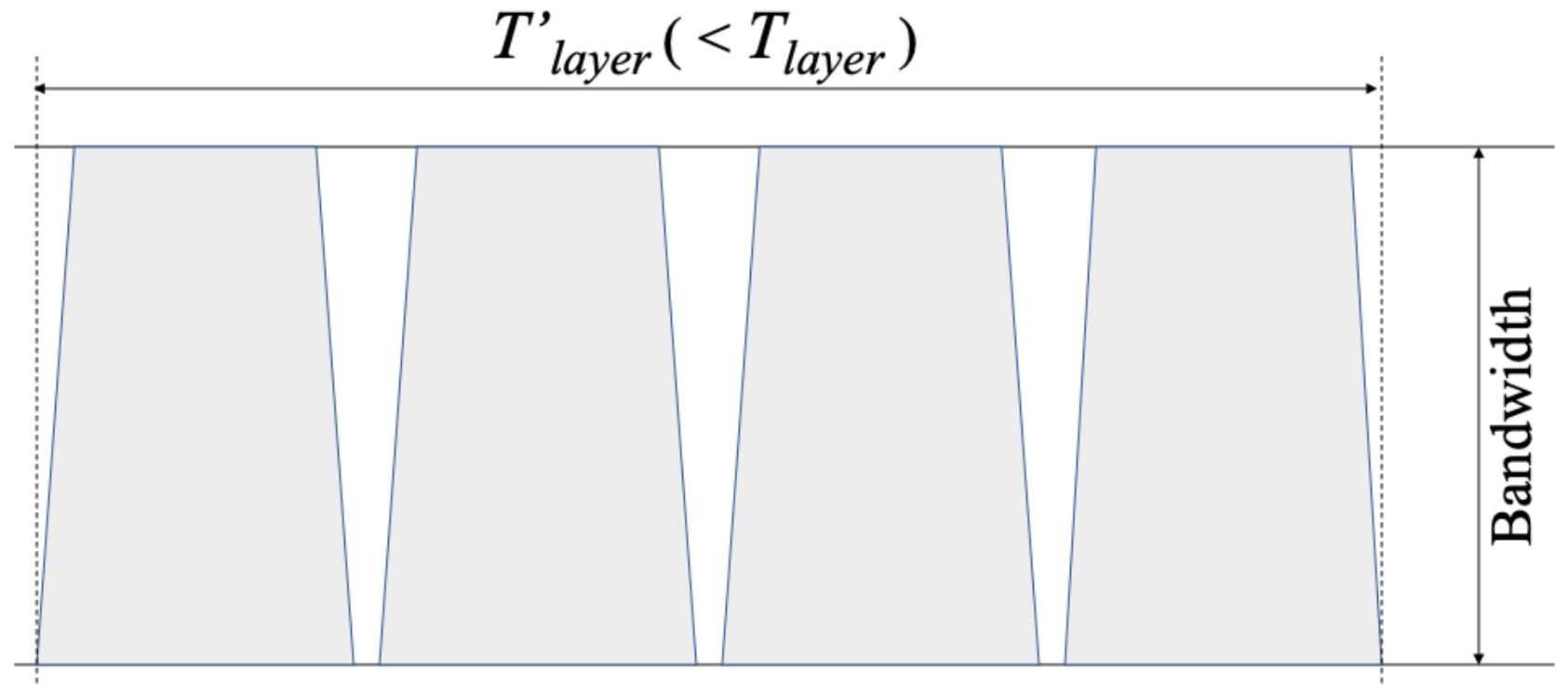}
	}
	\caption{Network traffic for different number of workers}
	\label{fg_band}
\end{minipage}
\end{figure}

\section{Shuffle-Exchange SGD}
In this section, we present the \emph{Shuffle-Exchange SGD} (SESGD). It greatly reduces the number of
handshakes while shares the same convergence performance with Ring-SGD.

\subsection{Shuffle-Exchange}

The number of handshakes can be greatly reduced if we cut down the number of workers. Intuitively, we divide all workers into
several groups and one worker only communicates with those in the same group. After the training, we average the global model
parameters among all groups.

However, in the case of DSGD, since different workers owns different datasets,
the model parameters of each worker will be slightly different after each iteration.
We need to ensure that every worker can communicate \emph{directly or indirectly} with all the other workers during the training process. 
\cite{d2_reduce1, localsgd} propose a hierarchical ring scheme for distributed training with massive nodes, which greatly reduces the number of handshakes. But this scheme also multiplies communication traffic in the same time as the additional allreduce and broadcast operations have to be introduced.

To solve this problem, after each iteration, we randomly shuffle these workers into different groups, 
which can be described in Figure \ref{fg_mr_reduce}.
For further explanation, imagine a DNN training with workers $\{0, 1, 2, 3\}$.
We divide them into two groups. In iteration $t$, workers in the two groups are $\{0, 1\}, \{2, 3\}$, but for the iteration $t'$, it may be $\{0, 2\}, \{1, 3\}$. 
During this operation, workers only need to change its sender and receiver, and the overhead of Shuffle-Exchange operation can be ignored. If we divide $n$ workers into $\sqrt{n}$ groups, the communication time of SESGD can be expressed as:
\begin{equation} \label{eq_sesgd}
T_{layer}  \approx \frac{2\mathcal{G}}{\nu}+2\sqrt{n}t_{\tau}
\end{equation}
The $O(n)$ handshakes are reduced to $O(\sqrt{n})$ comparing with Ring-SGD.
In our implementation, we use the pseudo-random algorithm to generate the grouping information and set the same random seed on every worker
to avoid extra message exchange and maintain the decentralized architecture.


\subsection{Gradient Correction}
Most distributed DL framework implements DSGD with the following formula:
\begin{equation}
\nabla_{t} = \frac{1}{n}\sum_{i=0}^{n-1}\nabla f_i \left(x_{t,i}; \xi_i \right) , \quad
x_{i, t+1}=x_{i,t}-\eta \nabla_{t} \label{norm_dsgd_1}
\end{equation}
where $\xi_i$ is the training sample from the local dataset $\mathcal{D}_i$ owned by worker $i$. $x_{i, t}$ is the local model parameters of worker $i$ in iteration $t$, and $\nabla_{t}$ denotes the global gradient.

However, the shuffle operation cannot be directly applied to \eqref{norm_dsgd_1}, since it ignores the gradient error among different groups. Let $g(i,t)$ be the group seeking function which return the group that worker $i$ belongs to in iteration $t$ and $\nabla_{g(i,t)}$ denotes the averaged gradient in group $g(i,t)$.
If the shuffle operation is directly applied to \eqref{norm_dsgd_1}, the gradient error in iteration $t$ will be abandoned once the next iteration begins, as shown in Figure \ref{fg_correction}.
Note all workers start with the same $x_0$, after $\tau$ iterations the model parameters in worker $i$ and worker $j$ will diverge and then harms the performance:
\begin{equation}
\mathbb{E}_{\xi \sim \mathcal{D}} \| x_{i,\tau} - x_{j,\tau} \| = \mathbb{E}_{\xi \sim \mathcal{D}}  \left\| \sum_{t \in [0, \dots, \tau]} [\nabla_{g(i,t)} - \nabla_{g(j,t)}] \right\|
\end{equation}

But if we regard the model parameters $x_{t}$ as ``gradient", then we can accumulate those gradient errors in local workers. Therefore, 
we can give the update formula of SESGD:
\begin{equation}
\hat{x}_{i,t} = x_{i,t} - \eta \nabla f\left(x_{t,i};\xi_{i}\right), \quad
x_{i,t+1} = \frac{k}{n} \sum_{j \sim g_{i, t}} \hat{x}_{j,t} \label{eq_mr_update_1}
\end{equation}
where ${x}_{i,t}$ is the model parameters in worker $i$ and $\hat{x}_{i, t}$ denotes the parameters that updated but not synchronized in worker $i$. 
SESGD can be described as Algorithm \ref{sesgd_algo}.

\begin{figure}[t]
\begin{minipage}{0.5\linewidth}
		\begin{algorithm}[H] \small
			\caption{Shuffle-Exchange SGD on worker $i$} 
			\label{sesgd_algo}
			\begin{algorithmic}[1] 
				\renewcommand{\algorithmicrequire}{\textbf{Input:}}
				\renewcommand{\algorithmicensure}{\textbf{Output:}}
				\REQUIRE Local dataset $\mathcal{D}_i$
				\REQUIRE Initialized parameters $x_{i,0}$
				\REQUIRE The number of iterations $T$
				\REQUIRE Minibatch size $b$
				\REQUIRE The number of workers $n$
				\REQUIRE The number of groups $k$
				\REQUIRE Random seed $\varsigma$
				\STATE Initialize the pseudo-random algorithm with $\varsigma$
				\FOR {$t = 0, 1, \dots, T-1$}
				\STATE $\widehat{x}_{i,t} \gets x_{i,t}$
					\FOR {$b=0, 1, \dots, b-1$}
						\STATE Sample data $\xi_i$ from local dataset $\mathcal{D}_i$
						\STATE Compute the gradients $\nabla f\left(x_{i,t}; \xi_i\right)$
						\STATE $\widehat{x}_{i,t} \gets \widehat{x}_{i,t} - \eta \frac{1}{b} \nabla f\left(x_{i,t}; \xi_i\right)$
					\ENDFOR

				\STATE Randomly generate new groups depending on $\varsigma$
				\STATE Enter a new group $G_{i, t}$
				\STATE $x_{i,t+1} \gets \text{Ring-AllReduce}\left(\widehat{x}_{i,t}; G_{i, t} \right)$
				\ENDFOR
				\STATE $\overline{x} \gets \text{Ring-AllReduce}\left(x_{i}  ;Global \right)$
			\end{algorithmic}
		\end{algorithm}
\end{minipage}
\begin{minipage}{0.45\linewidth}
		\centering
		~\\
			~\\
				~\\
					
		\subfigure[Shuffle-Exchange]{ \label{fg_mr_reduce}
			\centering
			\includegraphics[width=\linewidth]{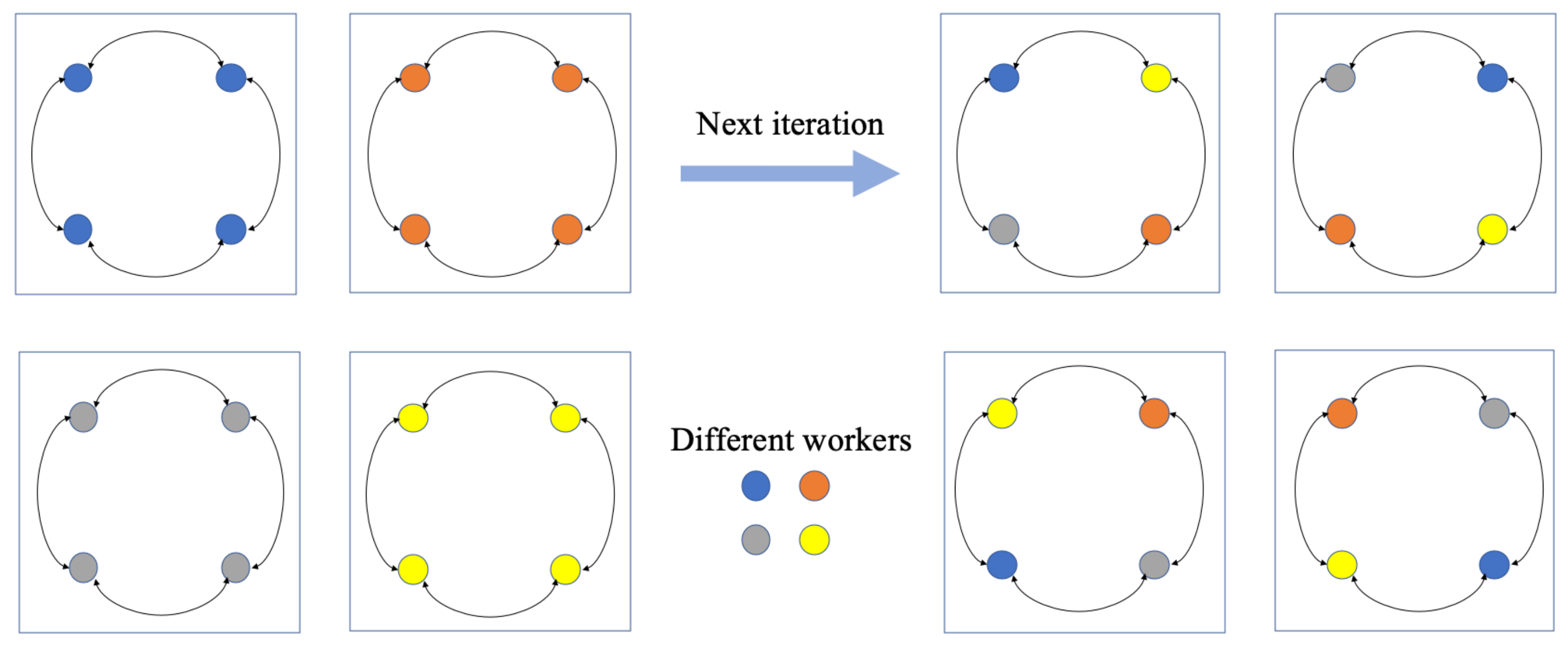}
		}
	
		\subfigure[Gradient Correction]{ \label{fg_correction}
			\centering
			\includegraphics[width=\linewidth]{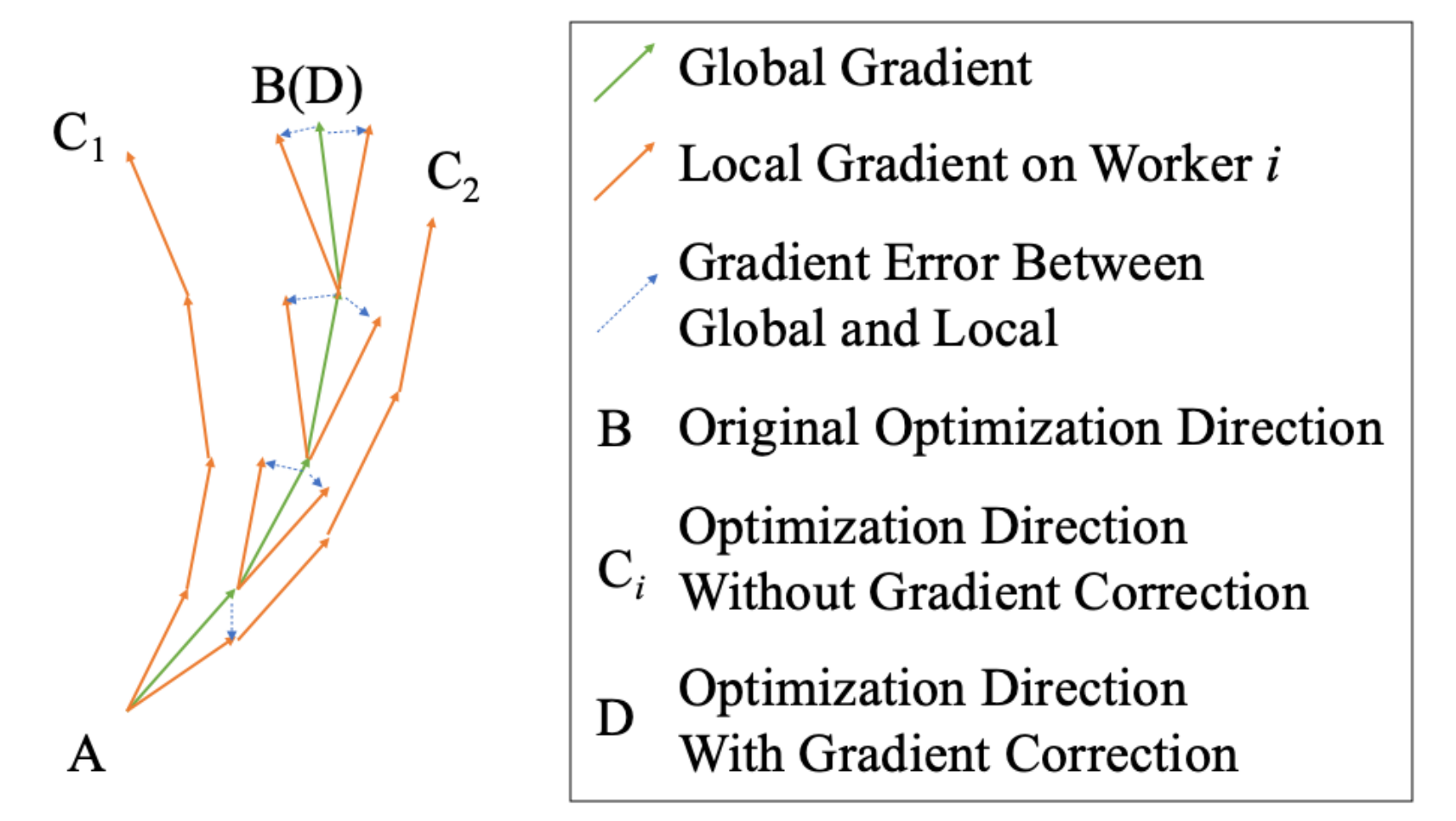}
		}
		\centering
		\captionof{figure}{Shuffle-Exchange SGD.}
\end{minipage}
\end{figure}


\subsection{Convergence Guarantee} \label{proof}
Here we provide convergence results in the smooth, non-convex functions. We make the following assumptions:
\begin{enumerate}
	\item (Unbiased) The stochastic gradients are unbiased estimators of the true gradient of the function $f$: $\mathbb{E}_{\xi \sim \mathcal{D}}[\nabla f(x)] = \nabla f(x)$. \label{unbiased}
	\item ($L$-smooth) There exists a constant $L > 0$ such that $\left\|\nabla f({x})-\nabla f({y})\right\| \leq L\|{x}-{y}\|$. \label{l_smooth}
	\item (Bounded gradient) It is standard to assume that the stochastic gradients over the sample space is bounded: $\mathbb{E}_{\xi \sim \mathcal{D}}\left\|\nabla f(x)\right\| \le M$. \label{bouned}
\end{enumerate}

For simplicity, we use $\bar{x}_t$ to denote $\frac{1}{n} \sum_{i=0}^{n-1} x_{i,t}$ and $\mathbb{E}[\cdot]$ for $\mathbb{E}_{\xi \sim \mathcal{D}} [\cdot]$. \cite{fixed_graph} define that a decentralized algorithm converges if for any $\epsilon > 0$, it eventually satisfies:
$\frac{\sum_{t=0}^{T-1} \mathbb{E}[\|\nabla f(\bar{x}_t)\|^2]}{T} \leq \epsilon$,
where $T$ is the number of iterations. We show that SESGD converges in this sense. Due to the page limit, we defer the derivation to supplement materials. 

\begin{lemma}
Suppose that all the assumptions holds and $k < n$. We have
\begin{equation}
\frac{\sum_{t=0}^{T-1} \mathbb{E}[\|\nabla f(\bar{x}_t)\|^2]}{T} \leq \frac{2(f(\bar{x}_0) - f^\star)}{\eta T} +  \frac{4\eta^2 L^2 M^2(nk-k)^2}{(n-k)^2} + \eta L M^2
\end{equation}
where $f^\star$ is the global minimum of optimization function $f$.
\end{lemma}

\begin{theorem}
	Given a success parameter $\epsilon>0$, consider iterations for non-convex function $f$. Set a fixed learning rate $\eta = \min\left\{\frac{\epsilon}{4LM^2}, \frac{\sqrt{\epsilon}(n-k)}{4(nk-k)LM}\right\}$, we have $\frac{\sum_{t=0}^{T-1} \mathbb{E}[\|\nabla f(\bar{x}_t)\|^2]}{T} \leq \epsilon$  for every iteration $T \ge \frac{4(f(\bar{x}_0)-f^\star)}{\eta \epsilon}$.
\end{theorem}

\section{Experiment}
In this section, we briefly discuss metrics and setup and review characteristics of those datasets,
and finally presenting experimental results and analysis.

\subsection{Experimental Settings}
\paragraph{Hardware.} 
We evaluate SESGD on a cluster with 4 nodes, wich are connected with 1Gbps Ethernet. The latency among nodes is about 0.1ms. Each node has 4 Nvida Tesla K80 GPUs, 2 Intel Xeon CPU E5-2660 cores and 128G memory. Each GPU is viewed as one worker in our experiment.

\paragraph{Software.}
We use PyTorch 1.3.1 and CUDA 10.1 to implement the algorithms in our experiments.
We implement the SESGD imitating the PyTorch \emph{DistributedDataParallel} module and use \emph{register\_hook} in PyTorch Tensor module to implement the overlap between computation and communication.

\paragraph{Datasets.}
We use three datasets for image classification.
\begin{itemize}
	\item CIFAR10 \cite{cifar10}: it consists of a training set of 50, 000 images from 10 classes, and a test set of 10, 000 images.
	\item CIFAR100 \cite{cifar10}: it is similar to CIFAR10 but has 100 classes.
	\item ImageNet \cite{deng2009imagenet}: the largest public dataset for image classification, includng 14.2 million labeled images from 1000 categories.
\end{itemize}

\paragraph{Tasks.}
We train ResNet18 \cite{resnet}, DenseNet121 \cite{densenet} and ResNet50 \cite{resnet} on 
CIFAR10, CIFAR100 and ImageNet separately.

\paragraph{Baseline.}
We compare SESGD with Ring-SGD and Local-SGD \cite{localsgd}, which updates parameters in local workers for several iterations, then
synchronize parameters with others to gain shorter training time. 
At the same time, we combine SESGD and Local-SGD (Local-SESGD) for further comparison.
All of those communication algorithms support computation and communication overlap.

\paragraph{Hyper-parameters.}
We use the following hyper-parameters.
\begin{itemize}
	\item Global batch size: 256 for both ResNet18 and DenseNet121, 512 for ResNet50.
	\item Learning rate: we start with learning rate from 0.1 and decay it a factor of 5 every 20 epochs for ResNet18 and DenseNet121. For ResNet50, the factor is 10 every 10 epochs.
	\item Weight decay: $5 \times 10^{-4}$ for ResNet18 and DenseNet121. $10^{-4}$ for ResNet50 as suggested in \cite{resnet}.
	\item Epoch: We run ResNet18 and DenseNet121 for 100 epochs. And for ResNet50 running on ImageNet, we run 30 epochs due to the time limit of hardware resources.
	\item Momentum: 0.9 for all the trainings.
	\item Group number and synchronization period: We set the group number to 4 for SESGD with 16 workers. The synchronization period of Local-SGD is set to 2.
\end{itemize}

\paragraph{Network Latency.}
We measure the average iteration time during training with different network latency to evaluate
the performances of these algorithms. We use Linux Traffic Control (tc) to control the network latency.


\setlength{\belowcaptionskip}{-0.5cm}
\begin{figure}[tb]
	\centering  
	\subfigure[ResNet18]{
		\centering
		\includegraphics[width=0.24\linewidth]{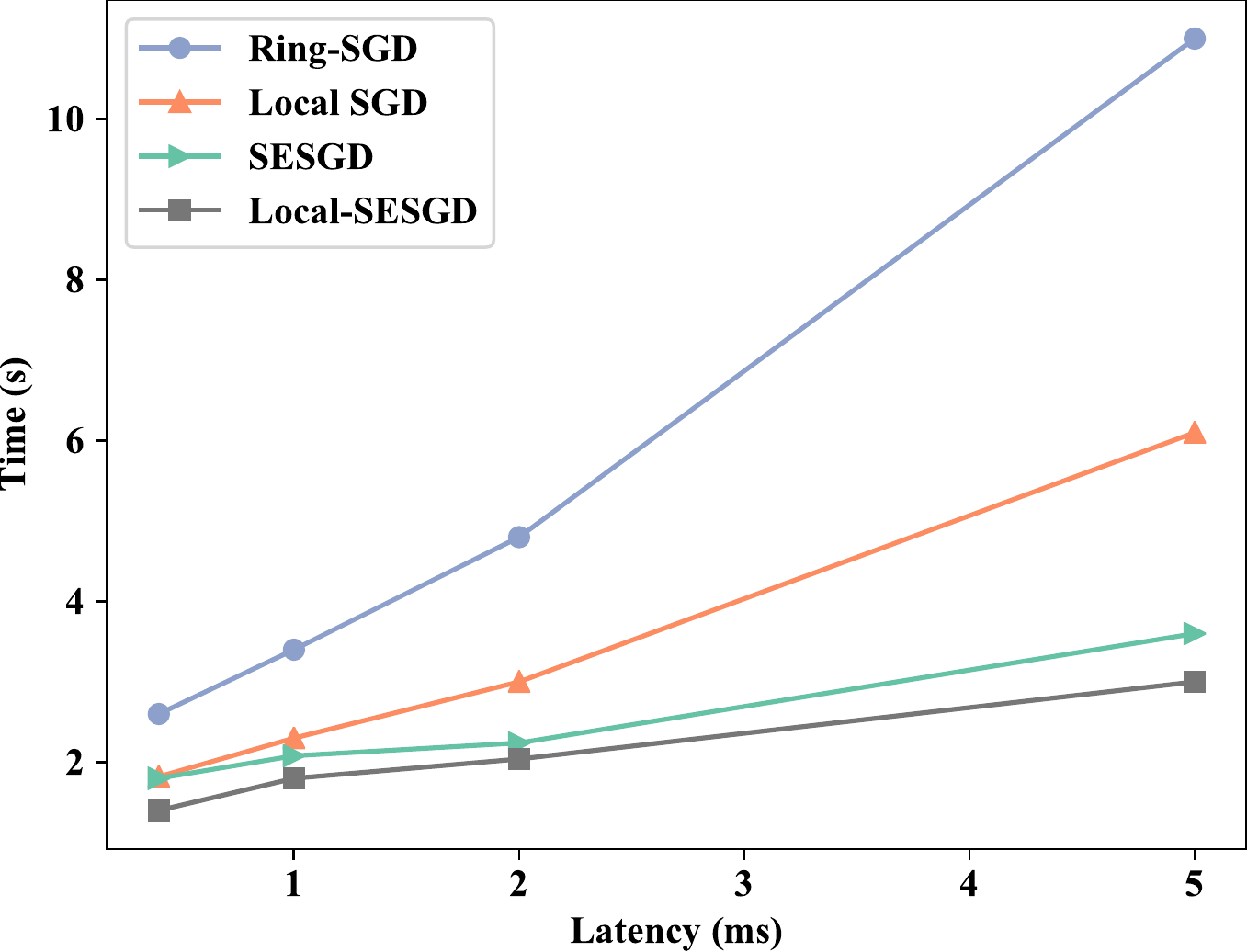}
		\label{fg-latency-ResNet18}
	}%
	\subfigure[DenseNet121]{
		\centering
		\includegraphics[width=0.24\linewidth]{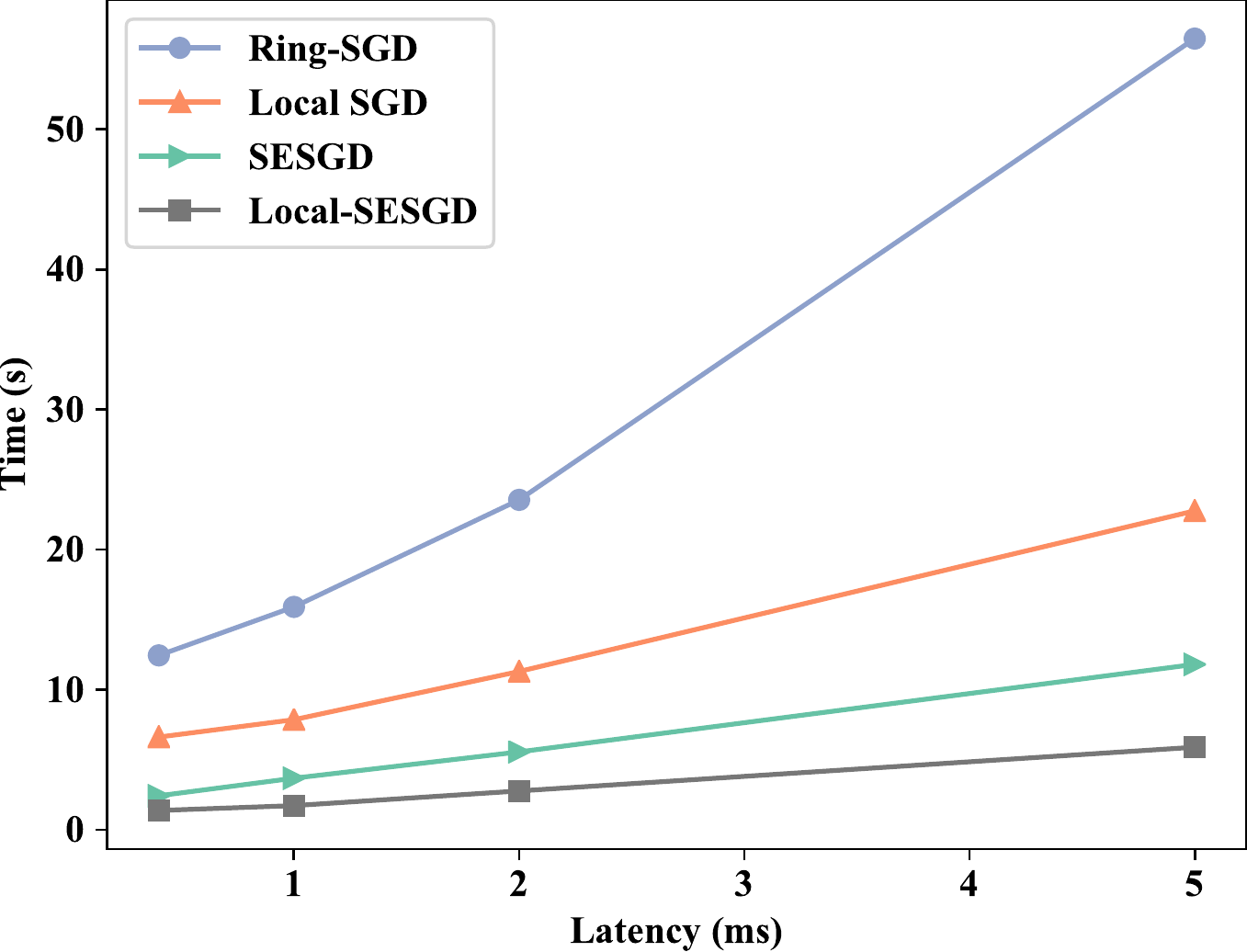}
		\label{fg-latency-DenseNet121}
	}%
	\subfigure[ResNet50]{
		\centering
		\includegraphics[width=0.24\linewidth]{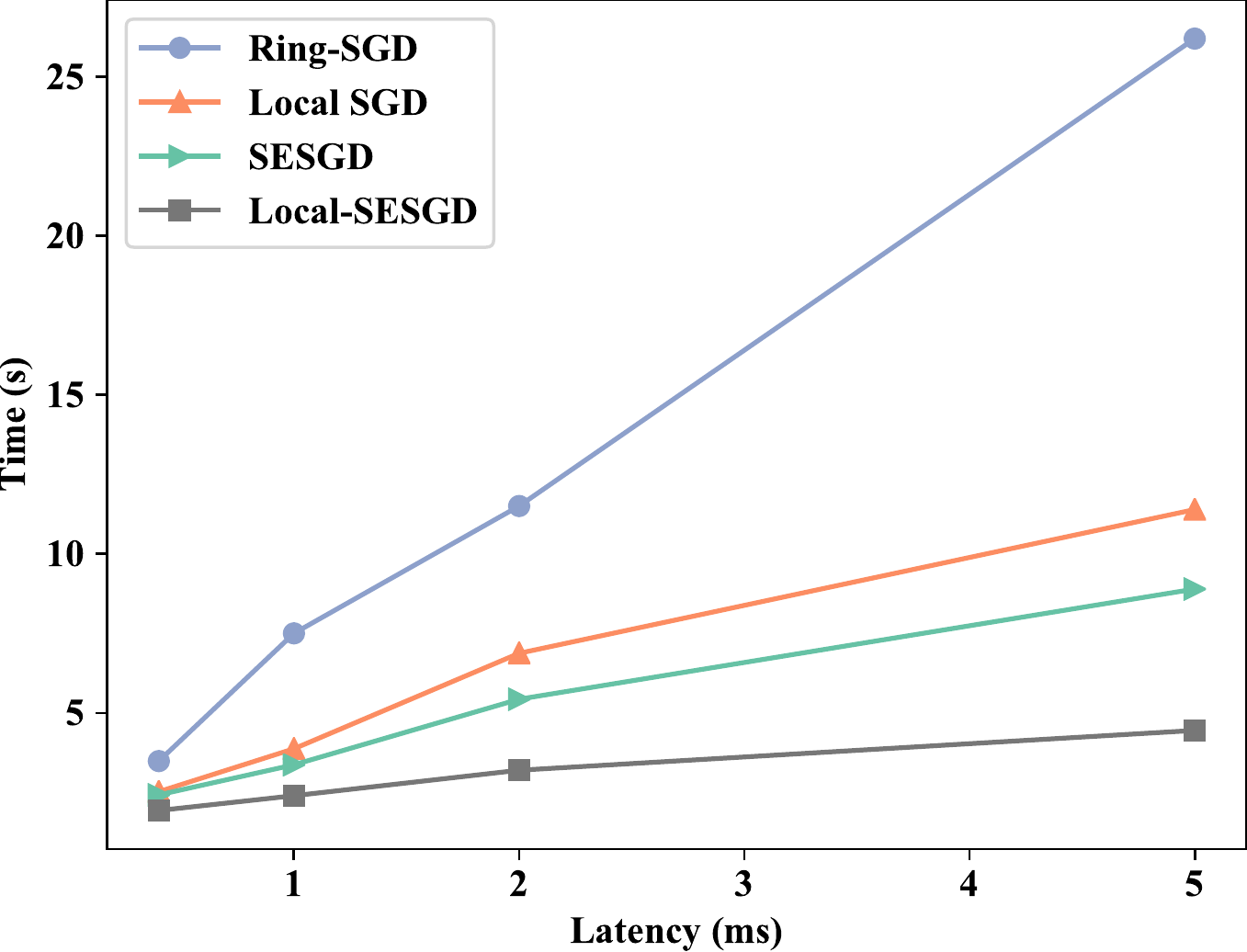}
		\label{fg-latency-ResNet50}
	}%
	\subfigure[VGG16]{
		\centering
		\includegraphics[width=0.24\linewidth]{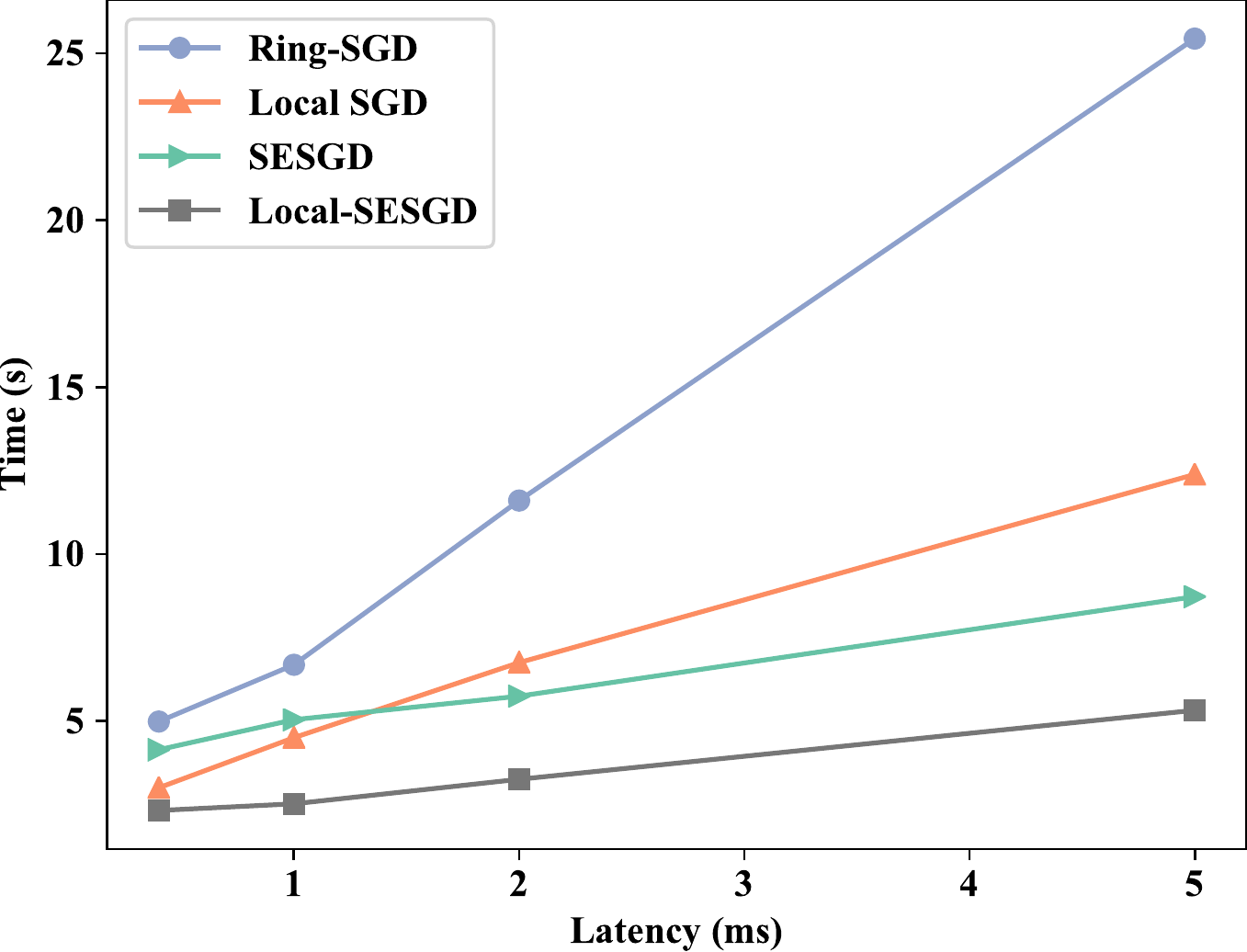}
		\label{fg-latency-VGG16}
	}%
	\centering
	\caption{These graphs plot the average per iteration time in different network latency with 16 workers. SESGD and Local-SESGD show good stability when the network latency grows.}
	\label{fg-latency}
\end{figure}

\setlength{\belowcaptionskip}{-0.5cm}
\begin{figure}[tb]
	\centering
	\subfigure[ResNet18 on CIFAR10]{  \label{fg_train_cifar10}
		\begin{minipage}[b]{0.3\textwidth}
			\includegraphics[width=\linewidth]{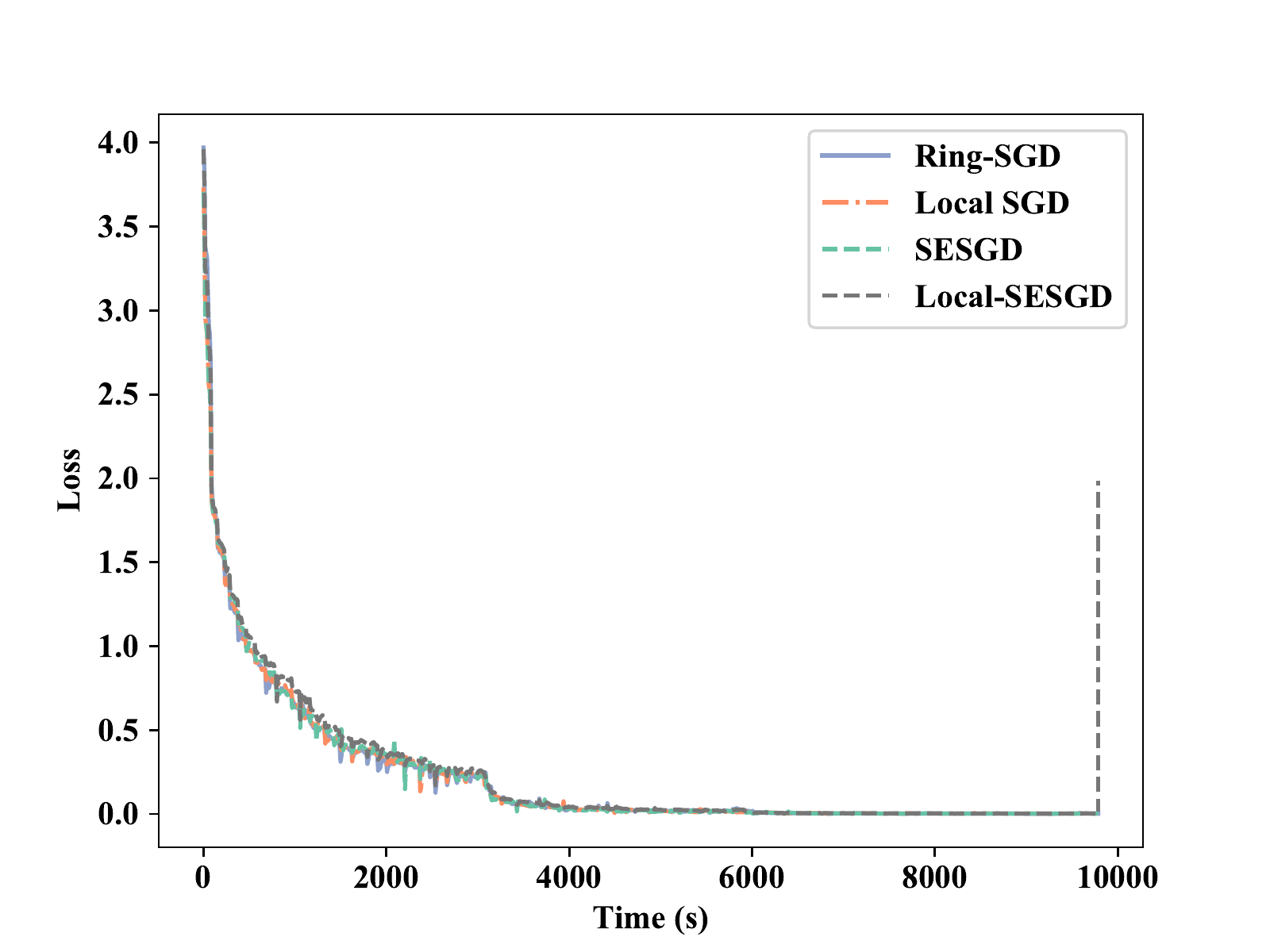}
			\includegraphics[width=\linewidth]{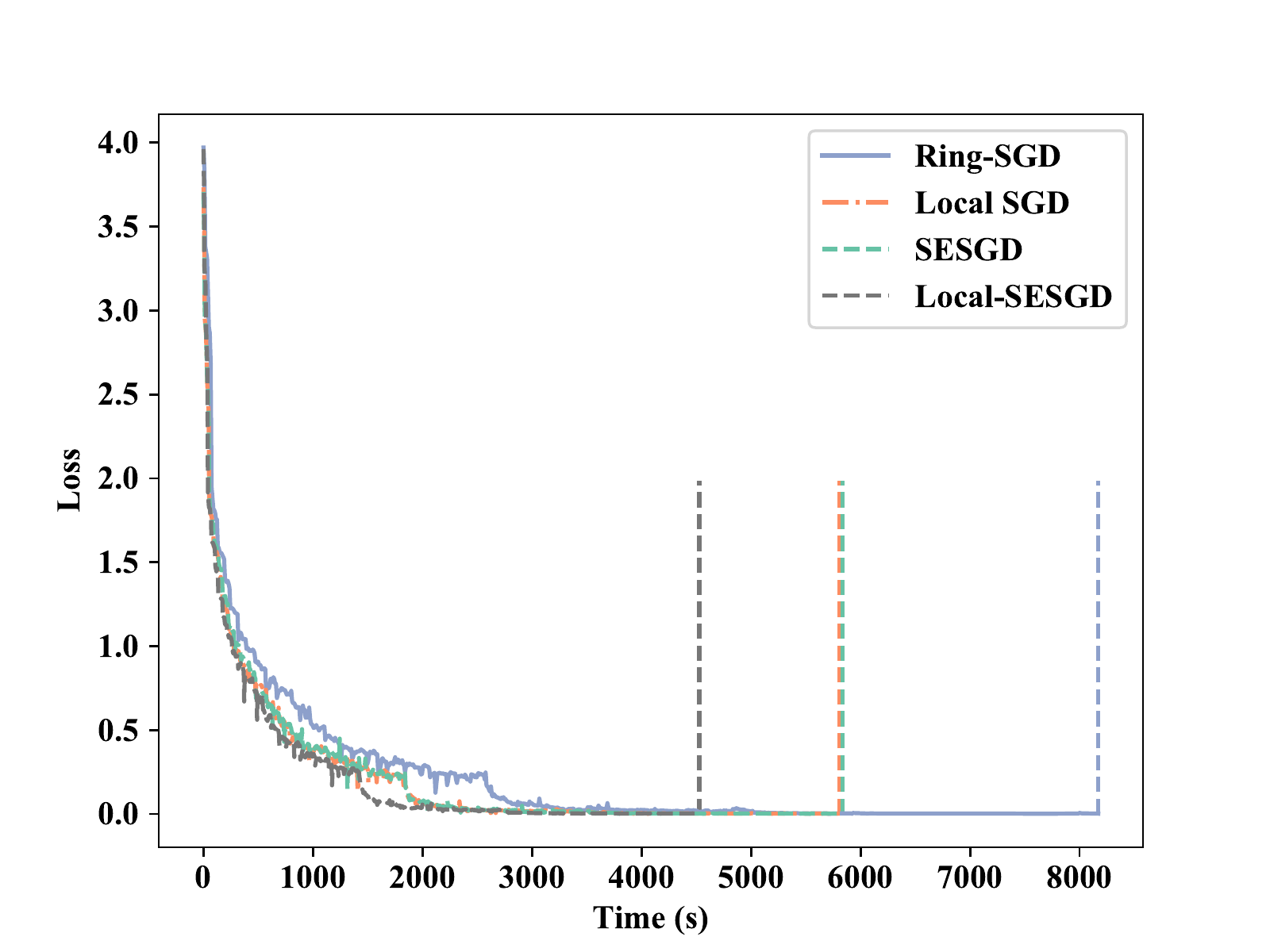}
			\includegraphics[width=\linewidth]{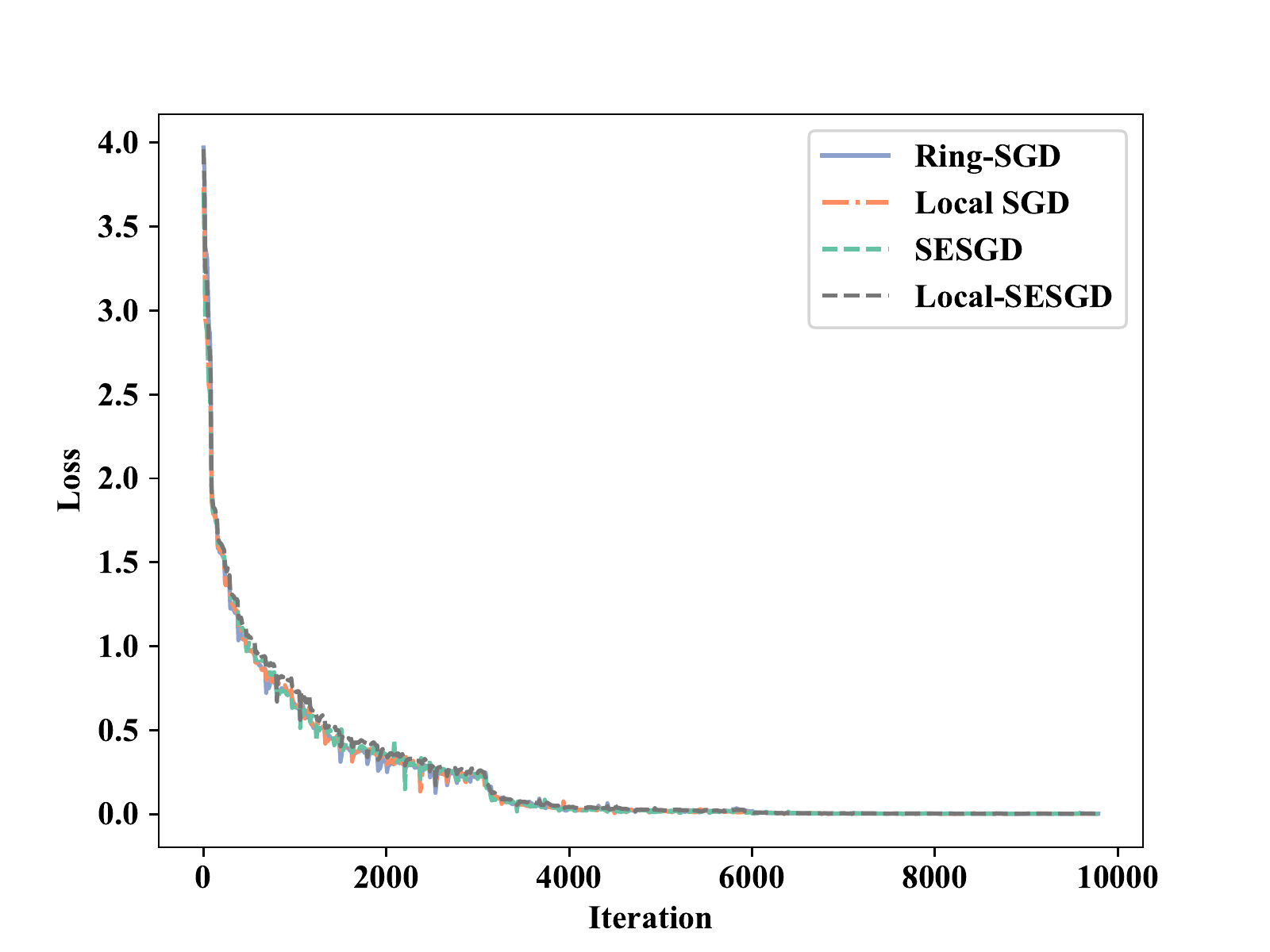}
		\end{minipage}	
	} 
	\subfigure[DenseNet121 on CIFAR100]{ \label{fg_train_cifar100}
		\begin{minipage}[b]{0.3\textwidth}
			\includegraphics[width=\linewidth]{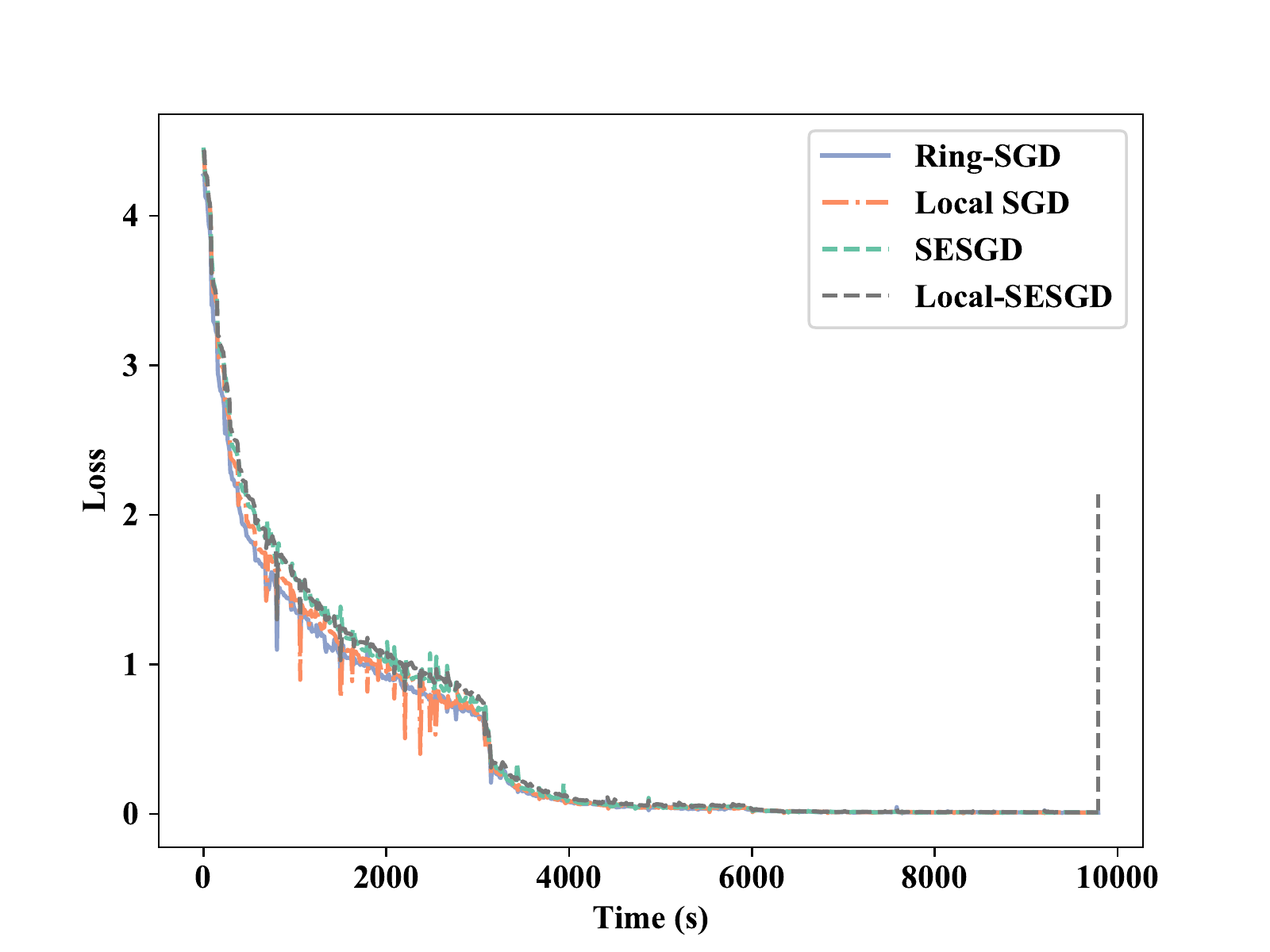}
			\includegraphics[width=\linewidth]{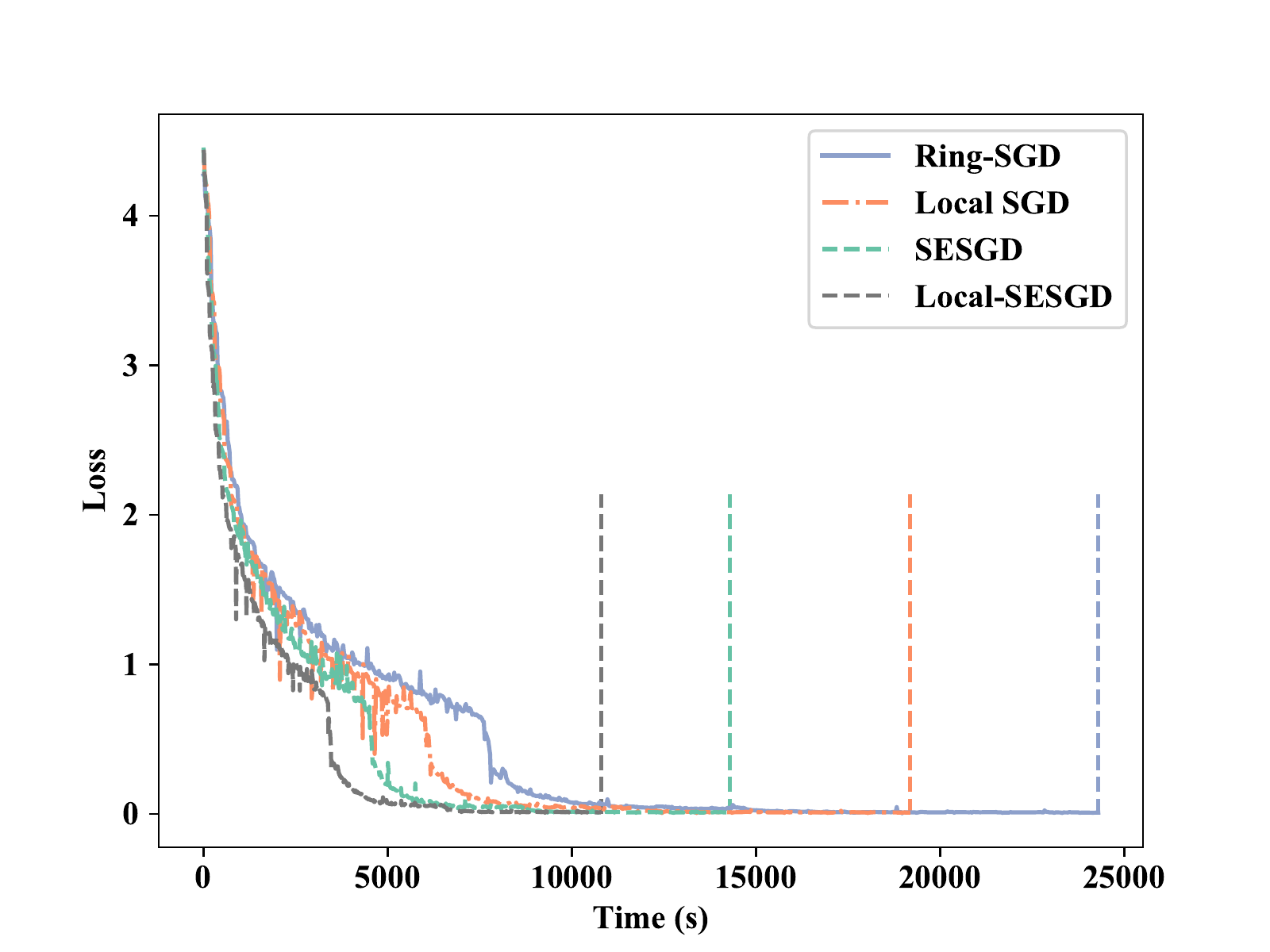}
			\includegraphics[width=\linewidth]{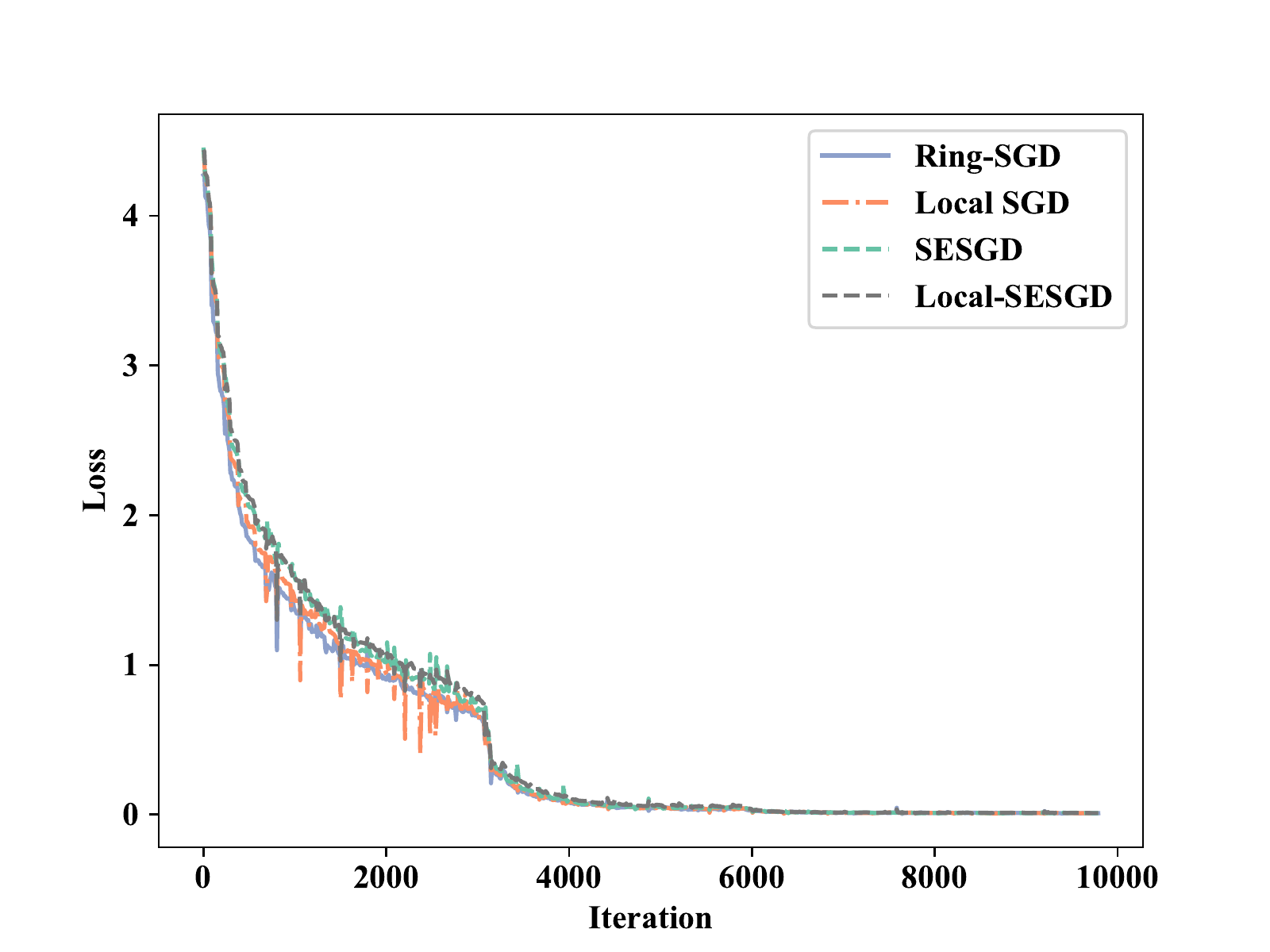}
		\end{minipage}	
	} 
	\subfigure[ResNet50 on ImageNet] { \label{fg_train_imagenet}
		\begin{minipage}[b]{0.3\textwidth}
			\includegraphics[width=\linewidth]{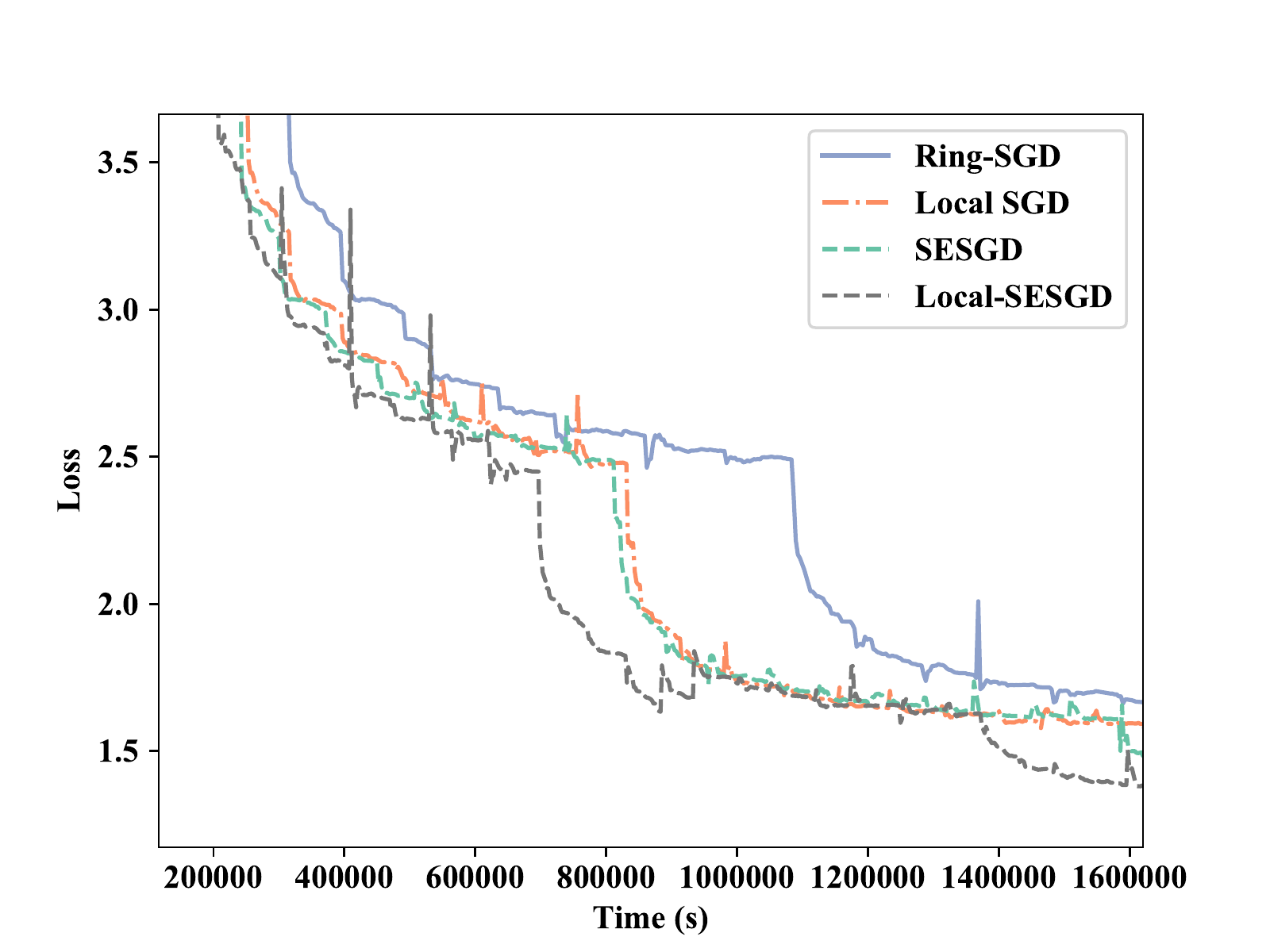}
			\includegraphics[width=\linewidth]{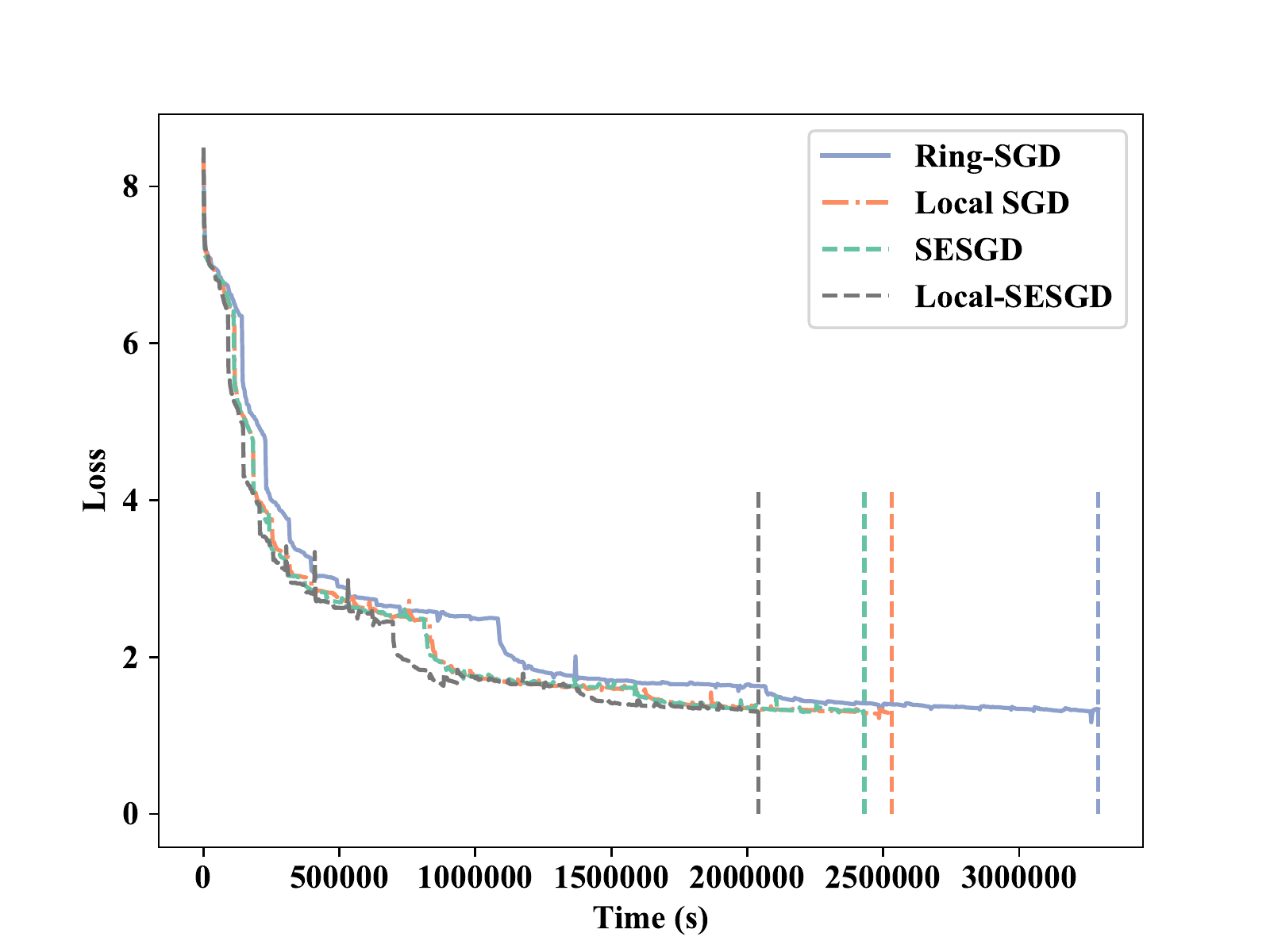}
			\includegraphics[width=\linewidth]{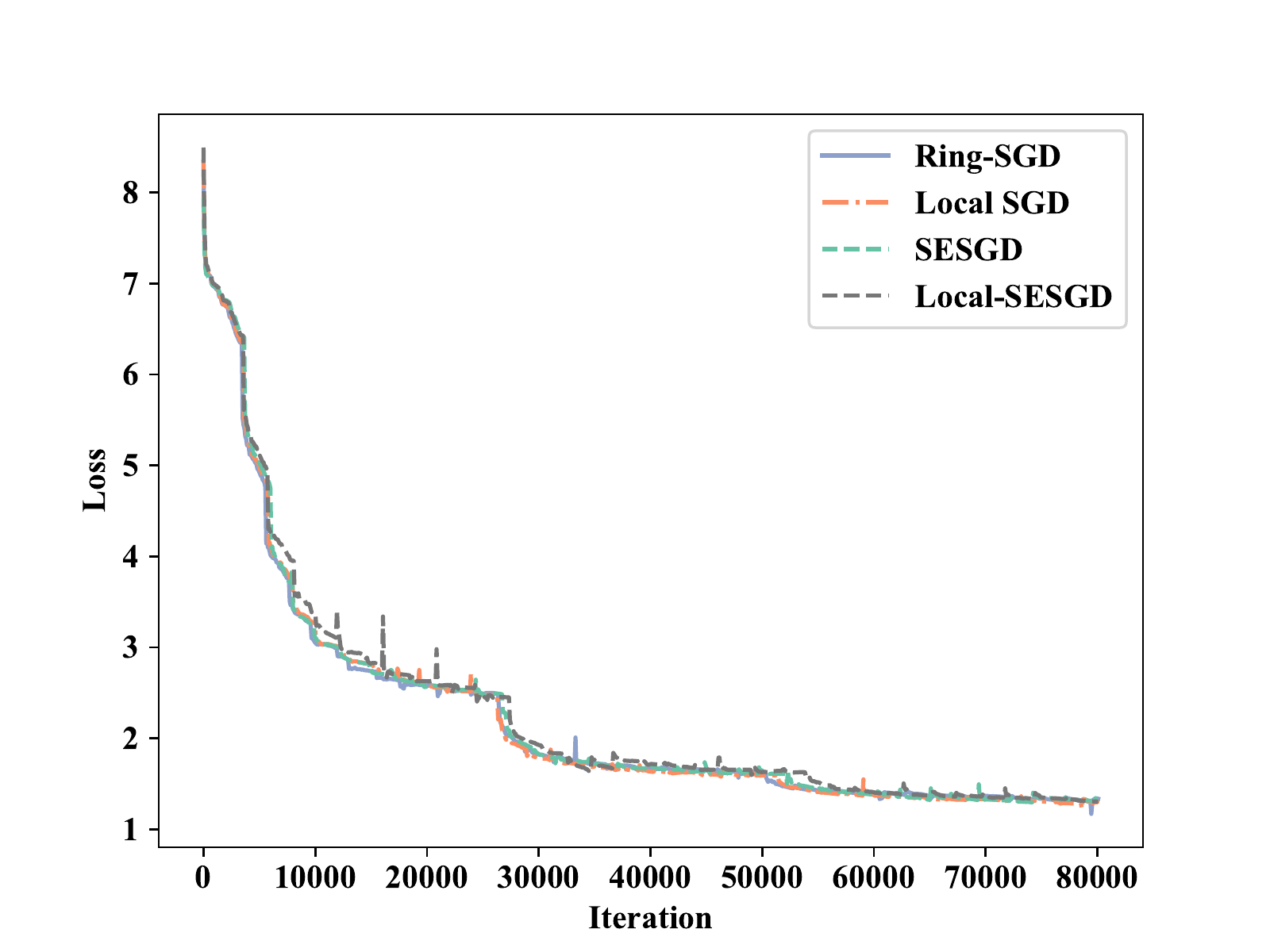}
		\end{minipage}	
	}
	\caption{Training in three different datasets with 0.1ms Ethernet latency. 
		The graphs on the top show the performances of these algorithms in specific periods. The medium parts show the overall training of DNNs per second while the bottom parts show the training loss per iteration.}
	\label{fg_time}
\end{figure}


\begin{figure}[htbp]
	\centering
	\subfigure[Speedup with 0.1ms latency] { \label{fg_speedup}
		\centering
		\includegraphics[width=0.48\linewidth]{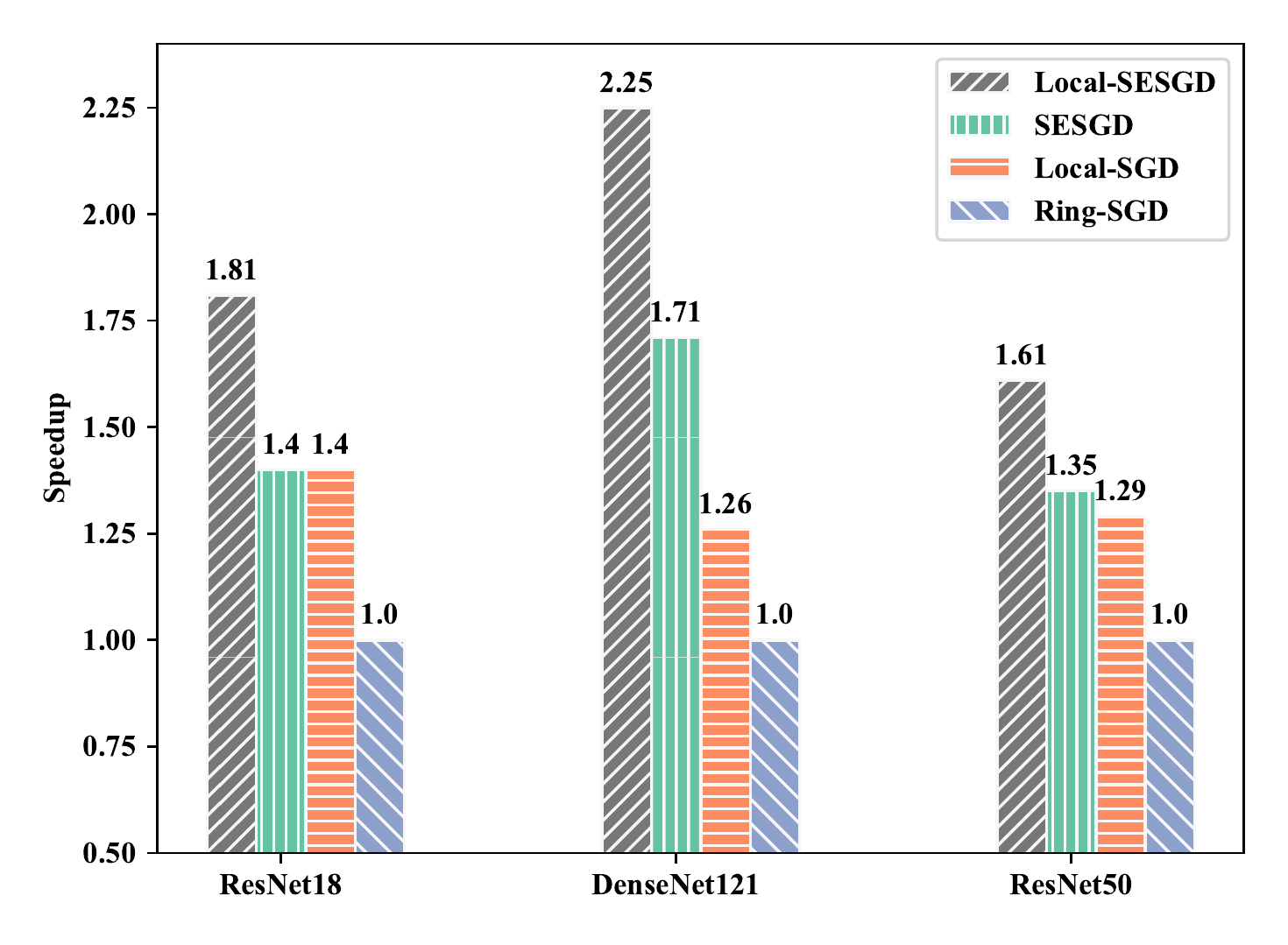} 
	}
	\subfigure[Model parameters distribution during training] { \label{fg_weight_distribution}
		\centering
		\includegraphics[width=0.48\linewidth]{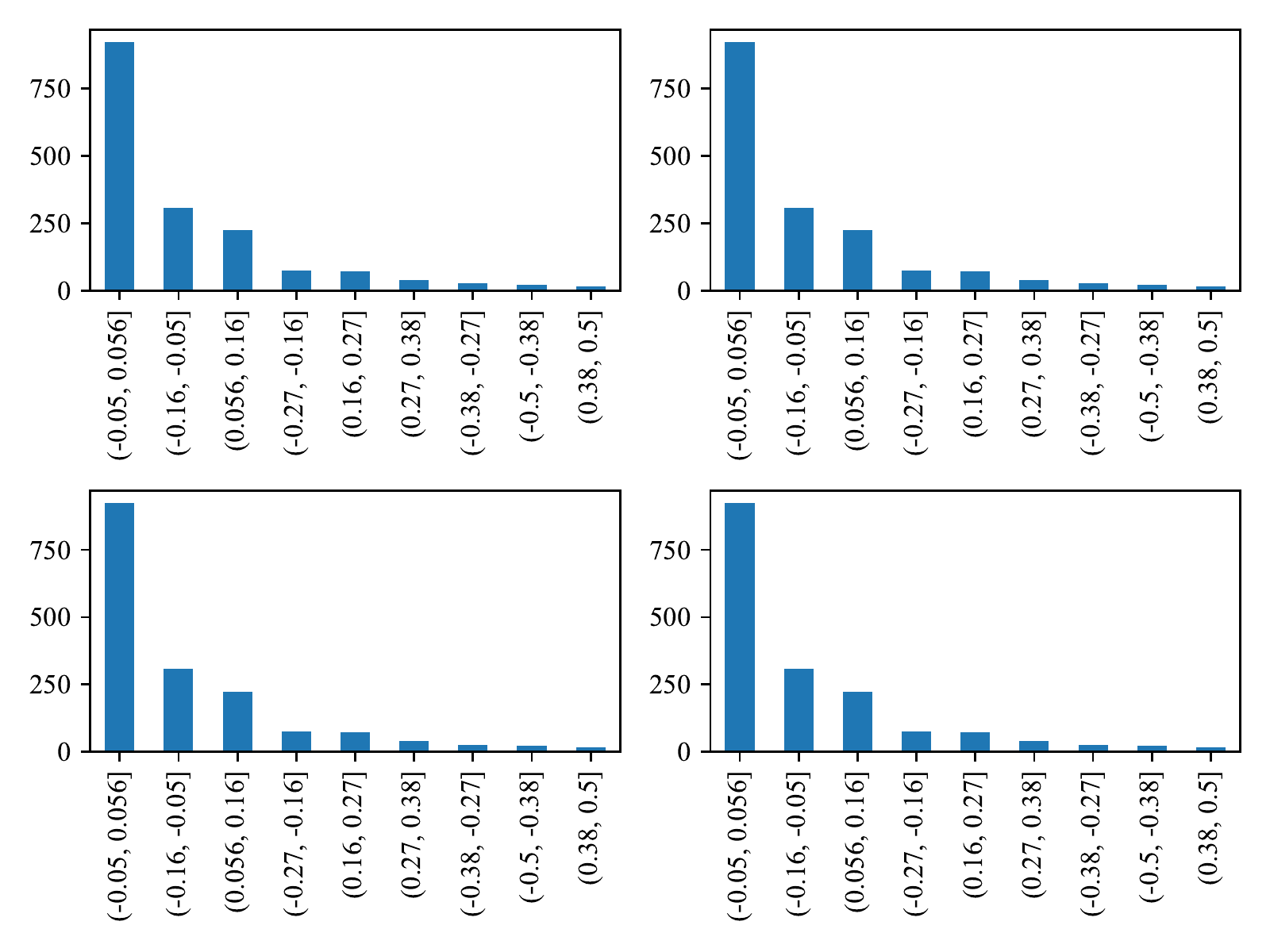} 
	}
	\caption{The left part shows the speedup comparing with different algorithms when training in 0.1ms Ethernet latency. And the right parts shows the parameters distribution in four different workers using SESGD the training progress The distribution is almost the same.}
\end{figure}

%

\subsection{Results Analysis}

\paragraph{Comparison of time speedup.}
Figure \ref{fg_time} plots the overall training time for different DNNs in three different datasets. And Figure \ref{fg_speedup} shows the final speedup for different algorithms. These experiments are all conducted with 0.1ms latency. SESGD achieves best among different algorithms and could combine with others to produce even better results. Comparing with the training time in Figure \ref{fg_train_cifar10} and Figure \ref{fg_train_cifar100}, SESGD shows an excellent performance in those models with many layers. In the training of DenseNet121, SESGD accelerates the training up to $1.7\times$.  Since SESGD makes previous efforts can be re-applied efficiently on large-scale training or with high network latency, $2.25 \times$ is achieved when combined with existing algorithms.

\paragraph{Comparison of convergence.}
\begin{wraptable}{r}{0.65\linewidth}
	\centering
	\begin{tabular}{lrrr}
		\toprule
		Algorithms      & CIFAR10 & CIFAR100 & ImageNet \\
		\midrule
		Ring-SGD   & 93.14\% & 73.69\%  & 71.168\% \\
		Local-SGD   & 92.97\% & 73.16\%  & 71.154\% \\
		SESGD       & 93.30\% & 73.62\%  & 71.266\% \\
		Local-SESGD & 93.02\% & 73.23\%  & 71.176\% \\
		\bottomrule
	\end{tabular}
	\caption{Final top-1 test accuracy  for three datasets}
	\label{tb_conv}
\end{wraptable}
During the training, the change of loss is quite similar for different algorithms, as shown in Figure \ref{fg_time}. Table \ref{tb_conv} lists the Top-1 Accuracy for different datasets when training with different algorithms. We can find that SESGD has the similar convergence performances on test datasets comparing with Ring-SGD, which confirms our proof in Section \ref{proof}.

\paragraph{Consistency of parameters.}
During training ResNet18 on CIFAR10 with SESGD, we randomly choose four workers and measure the parameter distributions in Conv1. Figure \ref{fg_weight_distribution} plots the result. 
Although we do not directly synchronize parameters among all workers, the model parameters of each worker maintain highly consistent.

\paragraph{Time cost in different latency.}
Figure \ref{fg-latency} shows the training time in different network latency. We run different DNNs in a range from 0.2ms (Ethernet) to 5ms (WLAN). As the latency increases, SESGD runs fastest and achieves almost $5 \times$ speedup against Ring-SGD in the WLAN latency. Local-SGD also performs better than Ring-SGD. But the handshakes during communication in Local-SGD remain unchanged, thus the communication time grows as the network latency increases. As Local SGD reduces the average number of parameters in every iteration and VGG16 has a large number of parameters and a smaller number of layers compared with others. Local SGD performs better than SESGD in shorter network latency. But in the case of large latency, idle time finally dominates. Thus SESGD performs best in the end.

\section{Conclusion}
In this paper, we observe the idle time caused by $O(n)$ handshakes in Ring-SGD harms distributed training performance and propose the Shuffle-Exchange SGD (SESGD) and guarantee its convergence. SESGD reduces the handshakes by dividing all workers in different groups and maintains the same performance with DSGD through Shuffle-Exchange and Gradient Correction operations. Our experiment shows SESGD can accelerate the DNN training by up to a range from $1.7\times$ to $5 \times$ in those networks with different latency.

\section*{Broader Impact}
Distributed SGD is widely adopt by those DL algorithms today. When training DNNs in datacenter, in order to reduce the communication overhead to gain shorter execution time, it is expensive to equip InfiniBand and Nvlink to get high bandwidth and low latency for training clusters. SESGD reduces the impact of network latency without introducing additional communication operations for ring based distributed training, which loose this constraint. When training DNN with those not-well-equipped cluster or in some specific applications (e.g., federated learning), SESGD can significantly reduce the communication delay without increasing the communication overhead or losing the accuracy.

%

\bibliographystyle{unsrt}
\bibliography{neurips_2020}

\newpage

\appendix

\section{Convergence Guarantee}

\begin{lemma}
	
	Suppose that all the assumptions holds and $k < n$. We have
	\begin{equation}
	\frac{\sum_{t=0}^{T-1} \mathbb{E}[\|\nabla f(\bar{x}_t)\|^2]}{T} \leq \frac{2(f(\bar{x}_0) - f^\star)}{\eta T} +  \frac{4\eta^2 L^2 M^2(nk-k)^2}{(n-k)^2} + \eta L M^2
	\end{equation}
	where $f^\star$ is the global minimum of optimization function $f$.
\end{lemma}

\begin{proof}
	Applying Descent Lemma:
	\[
	\begin{aligned}
	f(\bar{x}_{t+1}) & \leq f(\bar{x}_t) - \frac{\eta}{n} \nabla f(\bar{x}_t)^T \sum_{i=0}^{n-1} \nabla f_i(x_i; \xi_i) + \frac{\eta^2 L}{2 n^2} \left\| \sum_{i=0}^{n-1} \nabla f_i(x_i; \xi_i)\right\|^2  \\
	& = f(\bar{x}_t) - \eta \nabla f(\bar{x}_t)^T \nabla f(\bar{x}_t) + \eta \left( \nabla f(\bar{x}_t) - \frac{ \sum_{i=0}^{n-1} \nabla f_i(x_i; \xi_i)}{n}\right)^T  \nabla f(\bar{x}_t) \\
	& \quad + \frac{\eta^2 L}{2 n^2} \left\| \sum_{i=0}^{n-1} \nabla f_i(x_i; \xi_i)\right\|^2 
	\end{aligned}
	\]
	
	Taking expectation on randomness of samples and applying Cauchy-Schwarz inequality:
	\[
	\begin{aligned}
	\mathbb{E}[f(\bar{x}_{t+1})] & \leq f(\bar{x}_t) - \eta \| \nabla f(\bar{x}_t) \|^2 + \eta \mathbb{E} \left\| \nabla f(\bar{x}_t) - \frac{ \sum_{i=0}^{n-1} \nabla f_i(x_i)}{n}\right\| \|\nabla f(\bar{x}_t)\| \\
	& \quad + \frac{\eta^2 L}{2 n^2} \mathbb{E} \left\| \sum_{i=0}^{n-1} \nabla f_i(x_i)\right\|^2 \\
	& \leq f(\bar{x}_t) - \eta \| \nabla f(\bar{x}_t) \|^2 + \eta \mathbb{E} \left\| \nabla f(\bar{x}_t) - \frac{ \sum_{i=0}^{n-1} \nabla f_i(x_i)}{n}\right\| \|\nabla f(\bar{x}_t)\| \\
	& \quad + \frac{\eta^2 L M^2}{2}
	\end{aligned}
	\]
	
	For synchronous distributed update, we have $\nabla f(\bar{{x}}_t) = \frac{1}{n} \sum_{i=0}^{n-1}{\nabla f_i(\bar{x}_t)}$. The second last term can be bounded with the following inequality:
	\[
	\begin{aligned}
	\eta \mathbb{E} \left\| \nabla f(\bar{x}_t) - \frac{ \sum_{i=0}^{n-1} \nabla f_i(x_i)}{n}\right\|  &= \eta \mathbb{E}\left\| \frac{n\nabla f(\bar{{x}}_t) -  \sum_{i=0}^{n-1} \nabla f_i(x_i) }{n}  \right\| \\
	&= \eta\mathbb{E}\left\|\frac{1}{n} \sum_{i=0}^{n-1}\left[\nabla f_{i}(\bar{{x}}_t)- \nabla f_i(x_i)\right]\right\|\\
	& \leq \frac{\eta}{n} \sum_{i=0}^{n-1} \mathbb{E}\left\|\nabla f_{i}(\bar{{x}}_t)-\nabla f_{i}\left({x}_{i,t}\right)\right\|\\
	& \le \frac{\eta L}{n} \sum_{i=0}^{n-1} \mathbb{E}\left\|\bar{{x}}_t-{x}_{i,t}\right\|
	\end{aligned}
	\]
	
	The probability that two workers $i$, $j$ are in the different groups in one iteration is $\frac{n(k-1)}{k(n-1)}$ in a training cluster with $n$ workers and $k$ groups. If once worker $i$, $j$ are in the same group, the model parameters will be averaged.  Note $\bar{x}_0 = x_0$, suppose $k < n$, we have $\frac{n(k-1)}{k(n-1)} < 1$, then we have:
	\[
	\begin{aligned}
	\mathbb{E}\left\|\bar{{x}}_t-{x}_{i,t}\right\| & \le  \max_{j \in [n]}\{ \mathbb{E}\left\|x_{j,t}-{x}_{i,t} \right\|\}   \\
	& \le \eta  \max_{j \in [n]} \left\{ \sum_{t'=0}^{t} \left[\frac{n(k-1)}{k(n-1)}\right]^{t'}   \mathbb{E}\left\| \nabla_{g_{(j, t')}} - \nabla_{g_{(i, t')}} \right\| \right\}\\
	& \le 2\eta \sum_{t'=0}^{t} \left[\frac{n(k-1)}{k(n-1)}\right]^{t'}M \\
	& \le 2\eta \frac{nk-k}{n-k}M
	\end{aligned}
	\]
	The last inequality uses the sum of geometric series.
	
	So far we can further bound the above derivation. Note the inequality $ab \le \frac{a^2}{2} + \frac{b^2}{2}$ is used:
	\[
	\begin{aligned}
	\mathbb{E}[f(\bar{x}_{t+1})] & \leq f(\bar{x}_t) - \eta \|\nabla f(\bar{x}_t)\|^2 + \eta^2 \frac{2nk-2k}{n-k} LM \|\nabla f(\bar{x}_t)\| + \frac{\eta^2 L M^2}{2} \\
	& \leq f(\bar{x}_t) - \eta \|\nabla f(\bar{x}_t)\|^2 + \eta\left(\frac{2\eta^2 L^2 M^2 (nk-k)^2}{(n-k)^2} + \frac{\|\nabla f(\bar{x}_t)\|^2}{2} \right) + \frac{\eta^2 L M^2}{2} \\
	&= f(\bar{x}_t) - \frac{\eta \|\nabla f(\bar{x}_t)\|^2}{2} + \frac{2\eta^3 L^2 M^2 (nk-k)^2}{(n-k)^2} + \frac{\eta^2 L M^2}{2} 
	\end{aligned}
	\]
	
	Taking full expectation:
	\[
	\mathbb{E}[f(\bar{x}_{t+1}) - f(\bar{x}_t)] \leq  - \frac{\eta \mathbb{E}[\|\nabla f(\bar{x}_t)\|^2]}{2} +  \frac{2\eta^3 L^2 M^2 (nk-k)^2}{(n-k)^2} + \frac{\eta^2 L M^2}{2}
	\]
	
	For any $T > 0$, we have the following inequality:
	\[
	\begin{aligned}
	f(\bar{x}_0) - f^\star  & \ge \mathbb{E}[f(\bar{x}_0) - f(\bar{x}_{T})] \\
	& = \sum_{t=0}^{T-1}\mathbb{E}[f(\bar{x}_t) - f(\bar{x}_{t+1})] \\
	& \ge \frac{\eta \sum_{t=0}^{T-1} \mathbb{E}[\|\nabla f(\bar{x}_t)\|^2]}{2} -  \frac{2\eta^3 L^2 T M^2  (nk-k)^2}{(n-k)^2} - \frac{\eta^2 LT M^2}{2}
	\end{aligned}
	\]
	
	Move the term $\frac{\eta \sum_{t=0}^{T-1} \mathbb{E}[\|\nabla f(\bar{x}_t)\|^2]}{2}$ to the left, and divide both sides by $\frac{\eta T}{2}$:
	\[
	\frac{\sum_{t=0}^{T-1} \mathbb{E}[\|\nabla f(\bar{x}_t)\|^2]}{T} \leq \frac{2(f(\bar{x}_0) - f^\star)}{\eta T} +  \frac{4\eta^2 L^2 M^2(nk-k)^2}{(n-k)^2} + \eta L M^2
	\]
	Here complete the proof.
\end{proof}

\begin{theorem}
	Given a success parameter $\epsilon>0$, consider iterations for non-convex function $f$. Set a fixed learning rate $\eta = \min\left\{\frac{\epsilon}{4LM^2}, \frac{\sqrt{\epsilon}(n-k)}{4(nk-k)LM}\right\}$, we have $\frac{\sum_{t=0}^{T-1} \mathbb{E}[\|\nabla f(\bar{x}_t)\|^2]}{T} \leq \epsilon$  for every iteration $T \ge \frac{4(f(\bar{x}_0)-f^\star)}{\eta \epsilon}$.
\end{theorem}

\begin{proof}
	We bound the following inequality term by term:
	\[
	\frac{\sum_{t=0}^{T-1} \mathbb{E}[\|\nabla f(\bar{x}_t)\|^2]}{T} \leq \underbrace{ \frac{2(f(\bar{x}_0) - f^\star)}{\eta T} }_a + \underbrace { \frac{4\eta^2 L^2 M^2(nk-k)^2}{(n-k)^2} }_{b1} + \underbrace{ \eta L M^2}_{b2}
	\]
	
	For \[\eta = \min\left\{\frac{\epsilon}{4LM^2}, \frac{\sqrt{\epsilon}(n-k)}{4(nk-k)LM}\right\}\], we have $b1 \leq \epsilon /4 $, $b2 \leq \epsilon /4 $.
	
	Consider $a$, if $a \leq \epsilon / 2$, or \[T \ge \frac{4(f(\bar{x}_0)-f^\star)}{\eta \epsilon}\], we have the convergence as \[\frac{\sum_{t=0}^{T-1} \mathbb{E}[\|\nabla f(\bar{x}_t)\|^2]}{T} \leq \epsilon\].
\end{proof}

\end{document}